%% file: draft/draft.tex
\documentclass[12pt]{article}

\usepackage[utf8]{inputenc}
\usepackage[T1]{fontenc}
\usepackage{lmodern}
\usepackage{microtype}

\usepackage{geometry}
\usepackage{setspace}
\usepackage{placeins}
\usepackage{comment}

\usepackage{amsmath,amssymb,amsthm}
\usepackage{dsfont}
\usepackage{bigints}

\usepackage{csquotes}
\usepackage{nameref}

\usepackage{graphicx}
\usepackage{subfigure}
\usepackage{booktabs}
\usepackage{float}
\usepackage{threeparttable}
\usepackage{multirow}
\usepackage{dcolumn}
\usepackage[table,xcdraw]{xcolor}
\usepackage{caption}
\usepackage{tikz}

\usepackage{tikz}
\usepackage{forest}
\usepackage{pdflscape}

\definecolor{linkblue}{RGB}{0,90,150}
\usepackage{xurl}
\usepackage[
  colorlinks=true,
  linkcolor=linkblue,
  citecolor=linkblue,
  urlcolor=linkblue
]{hyperref}

\usepackage[
  backend=biber,
  style=authoryear,
  giveninits=true,
  uniquename=false,
  maxcitenames=2,
  maxbibnames=99,
  date=year,
  doi=false,
  isbn=false,
  eprint=false,
  url=false,
  sorting=nyt
]{biblatex}

\DeclareFieldFormat{note}{\mkbibparens{#1}}

\AtEveryBibitem{%
  \clearfield{month}%
  \clearfield{day}%
  \clearfield{language}%
  \clearfield{issn}%
}
\AtEveryBibitem{\clearfield{urldate}}

\newtheorem{proposition}{Proposition}
\newtheorem{assumption}{Assumption}

\title{
\textbf{Buyer Commitment in Bilateral
Bargaining: The Case of Online
Japanese C2C Market}
\thanks{
I thank Aaron Barkley, Kalyan Chatterjee, Paul Grieco, Anthony Kwasnica, Tai Lam, Nihal Mehta, Karl Schurter, Ran Shorrer, Chloe Tergiman, and Yuta Toyama, as well as participants at APIOC2025, EARIE2025, ESWC2025, IIOC2025, Penn State, and Waseda, for helpful comments.
I thank Mercari, Inc.\ for providing the data.
This work was supported by JST ERATO Grant Number JPMJER2301.
This research is part of the results of Value Exchange Engineering, a joint research project between Mercari R4D Lab and RIISE (Research Institute for an Inclusive Society through Engineering).\\
\textit{Email:} \href{mailto:kankuno@e.u-tokyo.ac.jp}{kankuno@e.u-tokyo.ac.jp}
}
}
\author{Kan Kuno \\ University of Tokyo}
\date{\today}

\begin{document}
\maketitle

% Abstract Section
\begin{abstract}
This paper studies bargaining when buyers can continue searching for alternative sellers while negotiating, which limits their commitment to complete a transaction. Using transaction-level data from a Japanese online marketplace, I document frequent post-agreement nonpurchase and show that buyers who explicitly pledge immediate payment are more likely to have their offers accepted, renege less often, and complete transactions faster. I develop and estimate a dynamic bargaining model with buyer search and limited commitment. Counterfactuals that restrict search during bargaining show that sellers, especially those with higher valuations, benefit from the elimination of delays and walkaways and respond by raising list prices. This reduces buyer welfare by lowering the option value of search and increasing expected list prices. Platform welfare increases modestly because the rise in list prices outweighs the decline in negotiated trade and accepted counteroffer prices. Overall, total welfare increases.
\end{abstract}



% Table of Contents (optional)
% \tableofcontents

%\chapter{Buyer Commitment in Bilateral Bargaining: A Model of Bargaining and Search}

\section{Introduction}

Negotiation between a buyer and a seller often takes place alongside buyer search. In housing markets, buyers typically negotiate with a seller while simultaneously continuing to visit and make offers on alternative properties. Similarly, in car purchases, buyers may solicit and compare quotes from multiple dealers before committing to a transaction. In online marketplaces, where price negotiation has become increasingly common, buyers can both haggle with a given seller and continue searching for alternative listings offering better terms. Despite the prevalence of such environments, existing empirical studies of bargaining typically analyze each bargaining interaction in isolation. As a result, we lack systematic evidence on how the availability of search during negotiation shapes bargaining outcomes in real-world settings.

Search availability during bargaining can undermine buyer commitment. Depending on the bargaining protocol, a buyer may walk away during negotiations or even renege after reaching an agreement if a better outside option becomes available. This has emerged as a practical concern on several online marketplaces, including eBay, Etsy, Facebook Marketplace, and Mercari. For example, sellers on eBay have long reported frustration with buyers who fail to complete purchases after their counteroffers are accepted, resulting in delayed or unpaid transactions. In response, eBay introduced an immediate payment feature that automatically charges the buyer upon offer acceptance (\textcite{ecommercebytes2021ebay}). Despite the use of such commitment-enhancing policies, their welfare implications for buyers, sellers, and platforms remain poorly understood. Understanding these implications requires analyzing not only how individual bargaining outcomes change but also how platform-wide equilibrium objects, such as the distribution of listing prices, respond.

This paper develops and estimates a structural model of complete-information bargaining with search, and quantifies the welfare effects of requiring buyer commitment. In the model, a searching buyer randomly matches with a seller and engages in price negotiation. A key feature is that buyers may continue searching during bargaining, which can generate delay and even reneging after agreement. The model predicts endogenous sorting into commitment status: among buyers who choose to negotiate, higher-valuation buyers are more likely to commit to suspending search, and sellers optimally accept steeper discounts from such committed buyers. Lower-valuation buyers, in contrast, must offer a premium to compensate for their higher risk of reneging.

I use detailed offer level data from the Japanese online marketplace Mercari. I exploit a unique feature of the data that records all textual communication exchanged during price negotiations. Using this information, I find evidence of endogenous sorting into commitment status, reflected in some buyers explicitly pledging immediate payment. These buyers complete transactions with substantially shorter delays and are significantly less likely to renege than buyers who do not make such pledges. Consistent with the model's predictions, sellers are also more willing to accept lower counteroffers from buyers who signal commitment.

I then take the model to the data, focusing on used 128 GB iPhone 7 listings. I provide constructive identification arguments for the key model primitives, including buyer and seller valuations, buyer and seller arrival rates, and the buyer search cost. These identification results exploit the optimality of counteroffers together with observed buyer choice probabilities across sellers. At the benchmark value, the estimated flow search cost is 820 yen per day, which corresponds to about 7.4\% of the median buyer valuation.

In my counterfactual, I simulate the introduction of a buyer commitment policy by making search during bargaining prohibitively costly. The results show that a full commitment policy increases overall welfare by 137.7 yen per listing and increases platform welfare by 27.7 yen per listing. Sellers benefit on average, with gains increasing in seller valuation, because the elimination of search-induced delays that were concentrated among high-valuation sellers raises continuation values and lowers the opportunity cost of raising list prices and being declined. Buyers face higher expected list prices ex ante and lose the option value of searching during bargaining, which reduces buyer welfare across the distribution, with larger losses for higher-valuation buyers who were more likely to purchase at list prices in the baseline. The platform gains modestly because higher list prices raise commission revenue on posted-price transactions, and this more than offsets the decline in negotiated trade and accepted counteroffer prices.

\section{Literature}

This paper contributes to three strands of the literature. The first is the empirical analysis of bargaining. Recent work in this area makes use of the availability of detailed offer level data, particularly from online marketplaces such as eBay. These data allow researchers to study bargaining behavior at a fine level of detail, including comparisons of theoretical models (\textcite{backus2020sequential}), measurement of inefficiencies and bargaining breakdowns (\textcite{larsen2021efficiency}; \textcite{freyberger2025well}), the role of verbal communication (\textcite{backusCommunicationBargainingBreakdown2020}), fairness considerations (\textcite{keniston2021fairness}), and the effects of timing and delay in negotiation (\textcite{fong2025effects}). While most existing studies analyze bargaining interactions in isolation, this paper adopts a structural approach that embeds bargaining within a competitive market environment and explicitly models inter-bargaining dynamics. This perspective highlights the importance of equilibrium effects for interpreting observed offers and bargaining outcomes.\footnote{Methodologically, this paper is related to empirical work that studies dynamic strategic interaction on online marketplaces, including \textcite{adachiCompetitionDynamicAuction2016}, \textcite{hendricksDynamicsEfficiencyDecentralized2018}, \textcite{backus2025dynamic}, and \textcite{bodoh-creedHowEfficientAre2021}.}

Secondly, this paper contributes to the empirical literature on platform design. Existing studies examine a wide range of design features, including search algorithms (\textcite{fradkinSearchMatchingRole2017}; \textcite{dinerstein2018consumer}), pricing mechanisms (\textcite{einavAuctionsPostedPrices2018b}), reputation systems (\textcite{resnickValueReputationEBay2006}; \textcite{noskoLimitsReputationPlatform2015}), and the role of user profiles (\textcite{edelmanRacialDiscriminationSharing2017}). In contrast, empirical work on the design of bargaining protocols remains limited. One notable exception is \textcite{zhangMeetMeHalfway2021}, who studies bargaining features on the Chinese platform Taobao and simulates the effects of a platform-wide ban on price negotiations. However, due to the absence of detailed offer level data, the bargaining component of their structural model is necessarily reduced form, which limits its ability to evaluate changes in the bargaining protocol beyond an outright ban. This paper instead focuses on a more practically prevalent design issue, namely unpaid items arising from bargaining interactions, and uses a structural approach to evaluate the consequences of alternative bargaining rules for buyers, sellers, and platform revenue.

Lastly, this paper contributes to the literature examining the interaction between search and bargaining. Much of the existing work is theoretical. For example, \textcite{chikteRoleExternalSearch1987} study a model without recall, implying that past offers cannot be revisited and bargaining delays do not arise, while \textcite{leeBargainingSearchRecall1994} and \textcite{chatterjeeBargainingSearchIncomplete1998} allow for recallable offers and show that buyer search can generate delays. Empirically, \textcite{allen2019search} study a negotiated price market in which consumers can search for competing offers, modeling multilateral bargaining as an auction conditional on search to quantify welfare losses from search frictions and price discrimination. \textcite{barkley2025haggle} study the residential housing market as a search environment with both auctions and price negotiations under two-sided incomplete information, modeling negotiation through a direct mechanism design approach. In contrast, my work abstracts from multilateral bargaining and two-sided incomplete information but explicitly models the bargaining protocol itself, allowing counterfactual analysis of practically relevant platform rules governing commitment, delay, and enforcement.

\section{Institution}

Mercari, Japan's largest online C2C marketplace for second-hand products, facilitates over 100 million transactions annually. The platform enables users to buy and sell a wide range of products, from handicrafts to automobiles. A seller creates a listing that includes photos, a description, and a chosen list price. Buyers may immediately purchase the product at the listed price on a first-come, first-served basis, or attempt to negotiate a price reduction. During my sample period from 2020 to 2021, negotiation was not a formal platform feature but instead took place publicly in the comment section of each listing. Figure \ref{fig:listing_ui} presents an illustrative listing page, and Table \ref{tab:comment_negotiation_example} provides a stylized example of a comment-based negotiation. When a seller and buyer agree on a price reduction in the comments, the seller updates the listing price accordingly. The buyer may then purchase the item at the revised price, but no mechanism binds the buyer to complete the transaction. If another buyer arrives, that buyer can purchase the item at the updated price, provided the seller does not revert the price after nonpurchase by the original buyer. Similar to eBay, post-agreement buyer reneging is a widespread issue on Mercari (\textcite{allabout2023mercari}).

% Preamble requirements:
% \usepackage{graphicx}
% \usepackage{tikz}
% \usetikzlibrary{arrows.meta,positioning}

% Fix: your TikZ/PGF version likely doesn't know the 'Latex' tip.
% Use the built-in arrow tip 'stealth' (or '->') and avoid arrows.meta.

% Preamble requirements:
% \usepackage{graphicx}
% \usepackage{tikz}
% \usetikzlibrary{positioning} % (optional)

% Preamble:
% \usepackage{graphicx}
% \usepackage{tikz}
% \usetikzlibrary{positioning}

\begin{figure}[!htbp]
\centering
\begin{tikzpicture}

\node[anchor=south west, inner sep=0] (image) at (0,0)
  {\includegraphics[width=0.50\linewidth]{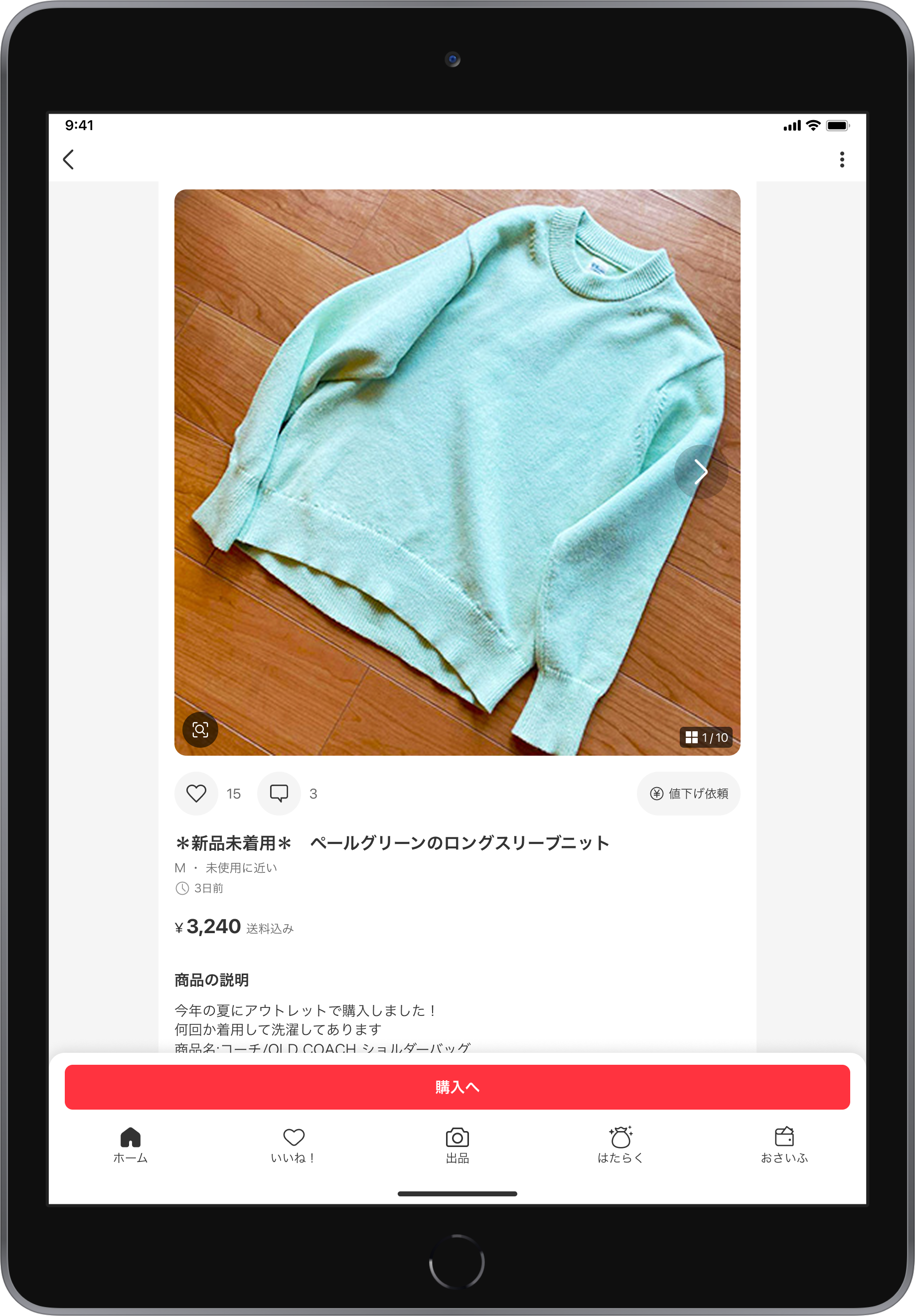}};

\begin{scope}[x={(image.south east)}, y={(image.north west)}]

\tikzset{
  callout/.style={
    draw,
    rounded corners,
    fill=white,
    align=left,
    inner sep=2.5pt,
    font=\footnotesize
  },
  arrow/.style={->, line width=0.7pt}
}

% ---- LEFT SIDE ----

\node[callout, anchor=west] (title) at (1.02,0.35) {Item title};
\draw[arrow] (title.west) -- (0.68,0.36);

\node[callout, anchor=east] (price) at (-0.02,0.3) {List price};
\draw[arrow] (price.east) -- (0.19,0.3);

% ---- RIGHT SIDE ----

\node[callout, anchor=east] (comments) at (-0.02,0.43) {Number of likes and comments};
\draw[arrow] (comments.east) -- (0.20,0.40);

\node[callout, anchor=east] (desc) at (-0.02,0.20) {Item description};
\draw[arrow] (desc.east) -- (0.18,0.23);

\node[callout, anchor=west] (purchase) at (1.02,0.18) {Purchase button};
\draw[arrow] (purchase.west) -- (0.90,0.17);

\end{scope}
\end{tikzpicture}

\caption{Annotated listing interface with English labels.}
\label{fig:listing_ui}

\begin{flushleft}
\footnotesize
\emph{Note:} This figure shows an example listing page in the Mercari iPad application. The image is provided by Mercari and does not correspond to any particular listing. A small on-screen indication of price negotiation appears below the image on the right. During my sample period from 2020 to 2021, negotiation was not a formal platform feature and occurred through the public comment section, although comment-based negotiation remained possible even after formalization.
\end{flushleft}

\end{figure}

\begin{table}[!htbp]
\centering
\begin{threeparttable}
\caption{Illustrative negotiation via the public comment section (stylized example)}
\label{tab:comment_negotiation_example}
\begin{tabular}{@{}clp{0.72\linewidth}@{}}
\toprule
Timestamp & User & Comment \\
\midrule
2020 Nov 17 10:12 & Buyer A &
Hello. I am considering purchasing. If the device is unlocked and there is no SIM restriction, would you be willing to sell it for 12{,}000 yen? \\
2020 Nov 17 10:18 & Seller &
Thank you for your message. Yes, I can accept 12{,}000 yen. The device is unlocked and has no SIM restriction. \\
2020 Nov 17 10:25 & Buyer A &
Thank you. \\
2020 Nov 17 10:31 & Seller &
I have updated the listed price to 12{,}000 yen. \\
\bottomrule
\end{tabular}
\begin{tablenotes}[flushleft]
\footnotesize
\item \textit{Notes:} This table is a stylized illustration of a typical negotiation in a public comment section. It does not relate to any particular listing.
\end{tablenotes}
\end{threeparttable}
\end{table}

\section{Model}\label{sec:model}

\subsection{Setup}

Time is continuous. Sellers and buyers have valuations for a homogeneous good distributed according to $F_S$ and $F_B$, respectively, both absolutely continuous with respect to Lebesgue measure. Buyers incur a common search cost $c$ per day, and all agents share a common continuous discount rate $r$. Each agent's valuation, the search cost, and the discount rate are common knowledge. There are $N_S$ sellers and $N_B$ buyers in the market, and the platform collects a fraction $t=0.10$ of each realized sale. The sequence of events and payoffs is depicted in Figure~\ref{fig:sequence_tree}. Consider a seller with valuation $s$, who first chooses a list price $p^0$. Buyers arrive to this seller according to a Poisson process with rate $\lambda_S$, and each arriving buyer's valuation is drawn from $F_B$. From the buyer's perspective, sellers arrive according to a Poisson process with rate $\lambda_B$. Upon matching with seller $s$, a buyer with valuation $b$ chooses one of three actions: purchase at the list price ($A$), make a counteroffer $p^1$ ($C$), or decline ($D$). If the buyer chooses $A$, trade occurs immediately at price $p^0$ between buyer $b$ and seller $s$. When a trade occurs at price $p$ between buyer $b$ and seller $s$, the buyer receives payoff $b-p$ and the seller receives payoff $(1-t)p-s$, and both agents exit the market. If the buyer chooses $D$, the match dissolves. The seller receives continuation value $U_S(s)$. The buyer resumes search after an exogenously given delay governed by a Poisson process with rate $\lambda_R$ and then obtains continuation value $U_B(b)$. I refer to $\lambda_R$ as the \textit{response rate} and assume $\lambda_R > \max{\{\lambda_B, \lambda_S\}}$.

If the buyer chooses $C$, the seller's response, either to accept the offer $p^1$ ($A$) or to decline it ($D$), arrives at rate $\lambda_R$. If the seller chooses $D$, the match is dissolved, with the seller receiving $U_S(s)$ and the buyer immediately resuming search with continuation value $U_B(b)$. If the seller chooses $A$, the buyer's opportunity to complete the purchase arrives at rate $\lambda_R$. At the time of choosing action $C$, the buyer also decides whether to continue searching for another seller ($S$) or not ($N$). If she chooses $S$, she continues to search and faces the arrival of a new seller with valuation $s' \sim F_S$ at Poisson rate $\lambda_B$, incurring flow search cost $c$ until arrival. In this case, if the original seller chooses $A$, the buyer retains the option to purchase from the original seller at price $p^1$ until the arrival of the new seller.\footnote{For simplicity, I assume that the arrival of the new seller always occurs after the seller's response and the buyer's purchase opportunity.} Upon the arrival of the new seller, the buyer either chooses to purchase from the original seller ($P$), in which case trade occurs at price $p^1$ between buyer $b$ and seller $s$, or to walk away to the new seller ($W$). If the buyer chooses $W$, her match partner becomes the new seller $s'$. Even in this case, with exogenously given probability $\kappa$, the original seller trades at price $p^1$ with a different buyer whose behavior lies outside the model. This depicts the reality that, as anecdotally reported on the platform, once the discount has been reflected, the product can be purchased by another buyer if the original offer maker does not act promptly. Otherwise, the seller becomes unmatched and receives continuation value $U_S(s)$.

From a match between a seller $s$ and a buyer $b$, the seller receives a continuation value $V_S(s,b)$ and the buyer receives $V_B(s,b)$. The continuation value of an unmatched seller with valuation $s$ is given by
\[
U_S(s) = \frac{\lambda_S}{\lambda_S + r} \int V_S(s,b)\, dF_B(b),
\]
while the continuation value of an unmatched buyer with valuation $b$ is
\[
U_B(b) = \frac{\lambda_B}{\lambda_B + r} \int V_B(s,b)\, dF_S(s) - \frac{c}{\lambda_B + r}.
\]

\input{figures/sequence_tree}

\subsection{Equilibrium}

I restrict attention to a subgame-perfect Nash equilibrium that is both stationary and in steady-state. Stationarity refers to the strategy profile: a player's equilibrium action depends only on its own valuation and, if matched, on the valuation of the currently matched counterpart. steady-state refers to the market environment: the valuation distributions \(F_S\) and \(F_B\), as well as the measure of active sellers and buyers \(N_S\) and \(N_B\), remain constant over time. In steady-state, market entry exactly offsets exit. Whenever a seller exits the market upon a completed trade, a new seller with the same valuation enters the market, and analogously for buyers. Under the Poisson meeting process, this implies a flow-balance condition equating total meeting rates on the two sides of the market:
\begin{equation}\label{eq:lambda_b_steady_state}
\lambda_B N_B = \lambda_S N_S    
\end{equation}
That is, the total rate at which sellers meet buyers equals the total rate at which buyers meet sellers.
In addition, I impose the following condition:
\begin{assumption}[Monotonicity and Lipschitz]\label{ass:monotonicity}
The buyer continuation value $V_B(s,b)$ is decreasing in $s$.
Moreover, $V_B(s,\cdot)$ and $V_S(\cdot,b)$ are Lipschitz continuous with Lipschitz constant one.
\end{assumption}
\begin{assumption}[Indifference]\label{ass:indifference}
Whenever a seller is indifferent between choosing \(A\) and \(D\) in response to a buyer's counteroffer \(p^1\), the seller chooses \(A\).
\end{assumption}

Given these assumptions, I solve for the equilibrium using backward induction. 
A buyer \(b\) whose \(p^1\) was accepted and matched to a new seller \(s'\) walks away if, and only if, \(V_B(s', b) \geq b - p^1\).
Given the monotonicity Assumption \ref{ass:monotonicity}, this can be equivalently stated as \(s' \leq s^*(b, p^1)\) where \(s^*\) is implicitly defined as \(V_B(s^*, b) = b - p^1\). 
As such, before matching with the new seller, the probability that buyer \(b\) walks away is \(F_S(s^*(b,p^1))\).

The seller accepts \(p^1\) when his expected payoff from accepting it weakly exceeds that from declining it.
In case the buyer chooses to search, such \(p^1\) satisfies
\begin{equation}\label{eq:seller_accept_noncommit}
\frac{\lambda_B}{\lambda_B + r}
\Bigl(
    (1 - \kappa) F_S\!\bigl(s^*(b,p^1)\bigr)\, U_S(s)
    + \bigl[1 - (1 - \kappa) F_S\!\bigl(s^*(b,p^1)\bigr)\bigr]\,
      \bigl((1-t)p^1 - s\bigr)
\Bigr)
\;\ge\;
U_S(s).
\end{equation}
Notice that the probability of a walkaway is scaled down by \(1 - \kappa\), reflecting the fact that even when the buyer walks away, the seller may still be able to trade at price \(p^1\) with another buyer, with an exogenously given probability \(\kappa\).
Let $p^1_S(s,b)$ denote the smallest $p^1$ that satisfies the seller acceptance condition \eqref{eq:seller_accept_noncommit} with equality. 
Although no closed form expression is available for $p^1_S(s,b)$, the following proposition establishes its existence provided that the continuation value is positive.
\begin{proposition}\label{prop:p1s_existence}
Suppose $U_S(s) > 0$. Then for any $s$ and $b$, there exists $p^1_S(s,b)$ satisfying \eqref{eq:seller_accept_noncommit} with equality.
\end{proposition}
\begin{proof}
See Appendix~\ref{proof:p1s_existence}.
\end{proof}

In case the buyer chooses not to search, \(p^1\) should satisfy
\begin{equation}\label{eq:seller_accept_commit}
\frac{\lambda_R}{\lambda_R + r}
\Bigl((1-t)p^1 - s\Bigr)
\;\ge\;
U_S(s)
\end{equation}
Let \(p^1_N(s)\) denote the \(p^1\) that satisfies above with equality, which is:
\begin{equation}\label{eq:committed_offer}
    p^1_N(s)=\frac{1}{1-t}\left[s+\frac{\lambda_R+r}{\lambda_R}\,U_S(s)\right]
\end{equation}
The following is the first key prediction of the model:
\begin{proposition}[Commitment premium]\label{prop:commit_premium}
\(p^1_S(s, b) > p^1_N(s)\).
\end{proposition}
\begin{proof}
    See Appendix \ref{proof:commit_premium}.
\end{proof}
\noindent The intuition for this is that a buyer who chooses to search needs to compensate for its own risk of walking away for the seller to accept her counteroffer.

I refer to the buyer action of choosing \(C\), \(S\), and submitting a counteroffer \(p^1_S(s,b)\) as making a \textit{non-committed} counteroffer, labeled \(CS\). Likewise, choosing \(C\), \(N\), and submitting a counteroffer \(p^1_N(s)\) is referred to as making a \textit{committed} counteroffer, labeled \(CN\). In either case, submitting any counteroffer above the stated amount is suboptimal for the buyer given Assumption \ref{ass:indifference}, and submitting any amount below leads the seller to choose \(D\) after a stochastic delay governed by rate \(\lambda_R\), in which case the buyer receives continuation value \(U_B(b)\). This payoff coincides with that obtained when the buyer chooses \(D\) upon matching. Accordingly, I relabel the outcome in which the buyer chooses \(C\) and is declined by the seller, regardless of the search choice, as action \(D\). Given this relabeling, buyer actions upon matching with seller \(s\) can be summarized as one of \(\{A, CN, CS, D\}\). Let \(V_B(s,b;\chi,p^0)\) denote buyer \(b\)'s continuation value conditional on matching with seller \(s\) posting initial offer \(p^0\) and taking action \(\chi \in \{A, CN, CS, D\}\). Then \(V_B(s,b;\chi,p^0)\) is given by:
\begin{equation}\label{eq:vbchi}
V_B(s,b;\chi,p^0)
=
\begin{cases}
b - p^0
& \text{if } \chi = A, \\[6pt]

\displaystyle
\left(\frac{\lambda_R}{\lambda_R + r}\right)^2
\bigl(b - p^1_N(s)\bigr)
& \text{if } \chi = CN, \\[10pt]

\displaystyle
\frac{\lambda_B}{\lambda_B + r}
\int
\max\!\left\{
b - p^1_S(s,b),
\;
V_B(s',b)
\right\}
\, dF_S(s')
-
\frac{c}{\lambda_B + r}
& \text{if } \chi = CS, \\[10pt]

\displaystyle
\frac{\lambda_R}{\lambda_R + r}\, U_B(b)
& \text{if } \chi = D.
\end{cases}    
\end{equation}

I denote \(\chi_B(s,b) \in \{A, CN, CS, D\}\) as the optimal action buyer \(b\) takes upon match with seller \(s\).
To establish the second main prediction of the model, I impose a regularity condition, which ensures that the non-committed counteroffer $p^1_S(s,b)$ does not decrease too steeply with the buyer valuation $b$ (Assumption \ref{ass:regularity_wow} in Appendix \ref{sec:regularity}). Then the following proposition shows that the model gives rise to an endogenous self selection into commitment status:
\begin{proposition}[Cutoff representation of buyer actions]\label{prop:partition}
Fix $s$. Suppose that $\mathcal{M}(s)\subseteq\{D,CS,CN,A\}$ is nonempty and that the regularity condition of Assumption \ref{ass:regularity_wow} holds. Then there exist cutoff values
\[
-\infty \le \tau^{CD}(s) \le \tau^{NS}(s) \le \tau^{AC}(s) \le +\infty
\]
such that the buyer's optimal action can be represented as:
\[
\chi_B(s,b)
=
\begin{cases}
D,  & b \le \tau^{CD}(s),\\
CS, & \tau^{CD}(s) < b \le \tau^{NS}(s),\\
CN, & \tau^{NS}(s) < b \le \tau^{AC}(s),\\
A,  & b > \tau^{AC}(s),
\end{cases}
\]
where cutoffs may coincide (so some regions may be empty) depending on $\mathcal{M}(s)$.
\end{proposition}
\begin{proof}
    See Appendix \ref{proof:partition}.
\end{proof}
\noindent The key intuition is as follows. Given a seller, buyers with the highest valuations have the highest opportunity cost of entering a time-consuming price negotiation or search process. They choose to purchase the product at face value. Those with the lowest valuations instead choose to find an even lower-valuation seller without entering any price negotiation. Those in the middle send a counteroffer, but among them, those with relatively higher valuations have a higher opportunity cost of waiting for the arrival of a new seller and, as such, choose not to search. Those with relatively lower valuations choose to search and are thus non-committed.

Finally, turning to the seller's choice of optimal list price $p^0(s)$, first note that seller $s$'s continuation value conditional on matching with a buyer $b$ and choosing list price $p^0$ is:
\[
V_S(s,b;p^0)
=
\begin{cases}
(1-t)p^0 - s
& \text{if } \chi_B(s,b;p^0) = A, \\[6pt]

\displaystyle
\frac{\lambda_R}{\lambda_R + r}\, U_S(s)
& \text{if } \chi_B(s,b;p^0) \in \{CN, CS\}, \\[10pt]

U_S(s)
& \text{if } \chi_B(s,b;p^0) = D.
\end{cases}
\]
Integrating the above, the seller chooses \(p^0\) to maximize:
\begin{equation}\label{eq:seller_maximization}
\begin{aligned}
U_S(s)
=
\max_{p^0}
\Biggl\{
\frac{\lambda_S}{\lambda_S + r}
\Biggl[
&\mathbb{P}\!\bigl(\chi_B(s,b;p^0)=A \mid s\bigr)\bigl((1-t)p^0 - s\bigr)
\\
&+
\mathbb{P}\!\bigl(\chi_B(s,b;p^0)\in\{CN,CS\}\mid s\bigr)
\frac{\lambda_R}{\lambda_R + r}\,U_S(s)
\\
&+
\mathbb{P}\!\bigl(\chi_B(s,b;p^0)=D\mid s\bigr)\,U_S(s)
\Biggr]
\Biggr\}.
\end{aligned}
\end{equation}
I denote the seller's optimal list price by \(p^0(s)\). 

I discuss two monotonicity results that facilitate the identification and estimation strategy. The first concerns the equilibrium list price \(p^0(s)\), which plays a central role in the estimation. Under the assumption that the buyer valuation distribution has increasing hazard rate, \(p^0(s)\) is monotone in the seller valuation. This allows me to use the observable list price \(p^0\) in place of the latent seller type \(s\) as the conditioning variable.
\begin{proposition}[Monotone list pricing]\label{prop:p0_monotonicity}
Let \(\mathcal{S}\subseteq\mathbb{R}\). Suppose that for all \(s\in\mathcal{S}\), the feasible buyer action set \(M(s)\) contains \(A\), \(CN\), and \(D\). Suppose further that \(F_B\) has increasing hazard rate. Then, for each \(s\in\mathcal{S}\), the seller's optimal list price \(p^0(s)\) is unique and strictly increasing in \(s\) on \(\mathcal{S}\).
\end{proposition}
\begin{proof}
See Appendix~\ref{proof:p0_monotonicity}.
\end{proof}

The second concerns the non-committed counteroffer $p^1_S(s,b)$ with respect to the buyer type $b$, and is used to identify the search cost $c$.
To establish this result, I impose another regularity condition that ensures that, as $p^1_S$ rises, the positive effect of higher surplus conditional on trade as the counteroffer rises dominates the negative effect of increased probability that the buyer walks away as the counteroffer rises (Assumption \ref{ass:regularity_sc} in Appendix \ref{sec:regularity}). This yields the following monotonicity result.
\begin{proposition}\label{prop:p1s_monotonicity}
Suppose Assumption\ref{ass:regularity_sc} holds and that $f_S\!\left(s^*(b,p^1_S(s,b))\right)>0$. Then the non-committed counteroffer satisfies
\[
\frac{\partial p^1_S}{\partial b}(s,b) < 0.
\]
\end{proposition}

\begin{proof}
See Appendix~\ref{proof:p1s_monotonicity}.
\end{proof}

\section{Data}

I use an anonymized dataset obtained from Mercari.
The dataset includes, for each item, an anonymized item identifier, anonymized seller and buyer identifiers, the item title and description, the item category, and the item condition (unused, almost unused, no noticeable scratches or stains, some scratches or stains, scratches or stains present, overall condition is poor), as well as shipping conditions (shipping time in days, prefecture of origin, shipping method, and shipping payment responsibility, either prepaid by the seller or paid on delivery).
In addition, for each item, the dataset contains a complete timeline of events that records which anonymized user performed which action, including listing the item, liking the item, posting a comment, changing the price, or purchasing the item, and the corresponding timestamp.
Details on the construction of the dataset are provided in Online Appendix \ref{sec:dataset_construction}.

% Product choice
In line with the homogeneous-good model, I focus on a single product category: the iPhone 7 (128GB).
This market definition is similar to those adopted in, for example, \textcite{adachiCompetitionDynamicAuction2016}, \textcite{hendricksDynamicsEfficiencyDecentralized2018}, and \textcite{backus2025dynamic}.
I further restrict the sample to items with the following condition categories: no noticeable scratches or stains, some scratches or stains, and scratches or stains present. 
This leaves me with 12,451 items listed between June 24, 2020 and December 30, 2021.

%Table \ref{tab:hednoic} reports the estimation results of the hedonic regression. %The homogeneous-good assumption is further supported by limited substitution to outside goods: among users who neither sell nor purchase an item within the market definition, X\% do not purchase any other smartphone on the platform during the sample period.

I extract offer amounts from buyer comments. Because sellers may adjust prices independently of buyer interaction, I define the effective list price as the posted price after any price changes that occur prior to the first buyer offer. An \textit{observable match} is defined as the arrival of a distinct buyer through a non-seller comment or a purchase. There may also be an unobservable match in which a buyer arrives without making a purchase or leaving a comment. Such cases, in which the buyer choice of $D$ is not observed, are accounted for in the estimation. At the observable match level, an offer is classified as accepted if the seller subsequently updates the posted price to the offered amount. Matches are labeled as $C$ if an offer is accepted, $A$ if the item is sold without any offer, and $D$ otherwise. For matches with an accepted offer, I define a walkaway indicator that equals one if the offering buyer does not complete the purchase.

To align the data with the homogeneous goods assumption, I control for observable within market heterogeneity by estimating a hedonic price regression of list prices on observed item characteristics. I then residualize both list prices and offer amounts by subtracting the predicted component from this regression and adding back the intercept, so that prices are expressed net of observable characteristics. The corresponding regression results are reported in Appendix~\ref{sec:hedonic_regression}.

Table~\ref{tab:item_summary_stats} reports summary statistics at the listing level using residualized prices. Nearly all items are eventually sold, which is consistent with the implications of the model. About 20\% of items are sold through bargaining, indicating that bargaining is prevalent in this market. On average, sellers match with 1.57 buyers per item, implying that sellers have a positive continuation value during negotiations. Table~\ref{tab:match_summary_stats} reports summary statistics at the match level. Buyer reneging is strikingly common, occurring in 37.4\% of accepted-offer matches on average.

\input{tables/item_summary_stats}
\input{tables/match_summary_stats}

The data exhibit patterns that support the remaining model assumptions. Among users who make at least one purchase, 91.4\% purchase only one product, and 95.2\% of sellers list only one product, supporting the unit transaction assumption. In addition, only 7.2\% of listings involve a buyer who reappears through a comment or a purchase after the arrival of another buyer, supporting the assumption that there are no multiple simultaneous buyer arrivals.

\section{Preliminary Evidence} \label{sec:prelim}

A key prediction of the model is that buyers self-select into commitment status. High valuation buyers commit to transacting with the current seller, while lower valuation buyers continue to search while bargaining (Proposition~\ref{prop:partition}). In the data, such commitment can manifest as buyers signaling seriousness by pledging immediate payment. Figure~\ref{fig:wordcloud_adverbs} illustrates this behavior using a word cloud constructed from adverbs used in buyer messages during price negotiations, where language associated with immediacy, such as ``soon'' and ``immediate'', appears frequently. I collect 20 expressions related to immediacy and classify a buyer as making a pledge of \textit{immediate payment} if any of these expressions appears in the negotiation messages (see Online Appendix \ref{sec:dataset_construction} for the list). Overall, expressions indicating immediate payment appear in 41.4\% of buyer offer messages, suggesting that verbal commitment through payment timing is common in practice.

\begin{figure}[htbp]
\centering
\includegraphics[width=0.75\textwidth]{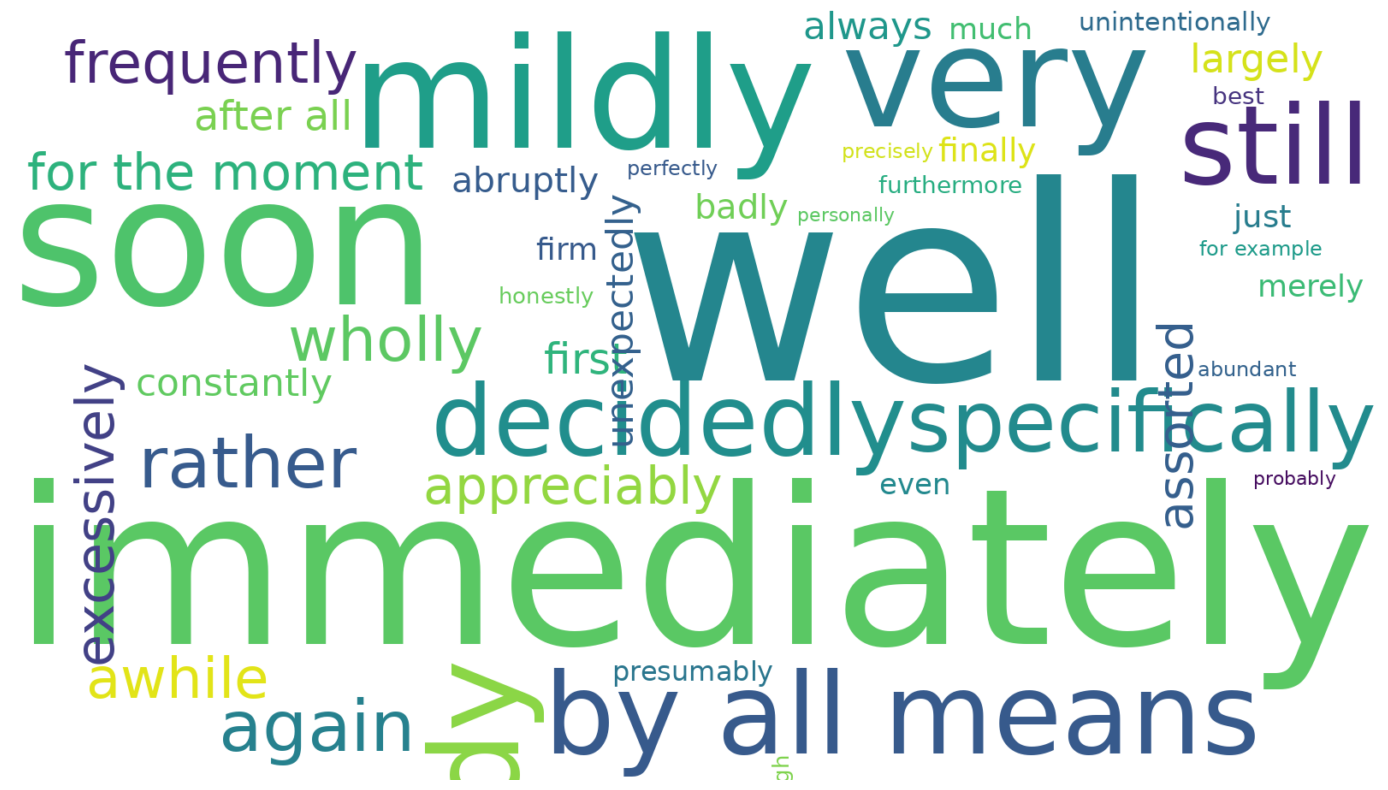}
\caption{Word Cloud of Adverbs Used in Bargaining Messages}
\label{fig:wordcloud_adverbs}
\vspace{0.5em}
\begin{minipage}{0.9\textwidth}
\footnotesize
\textit{Notes:} The figure shows an English word cloud constructed from adverbs appearing in buyer comments during bargaining matches. 
Japanese comments are tokenized using the Sudachi morphological analyzer. 
Tokens are filtered to exclude punctuation, single-character kana, and personal names. 
Japanese lemmas are translated into English using the Open Multilingual WordNet, mapping each Japanese lemma to the first available English lemma. 
Word sizes are proportional to token frequency, after aggregating frequencies across Japanese words that translate to the same English term and dropping very rare words.
\end{minipage}
\end{figure}

One might argue that pledges of immediate payment are merely cheap talk and do not reflect true commitment. Table~\ref{tab:commitment_walkaway_time} evaluates this claim using OLS regressions that relate bargaining outcomes to price concessions and commitment language in buyer messages. The independent variables include the discount, defined as the difference between the list price and the offer amount divided by the list price, and an indicator for whether the buyer pledged immediate payment in the offer message. The dependent variable is an acceptance indicator in column~1, a walkaway indicator in column~2, and $\log(1+\text{purchase time})$ in column~3, where time is measured in days. Column~1 uses 5,913 observed matches in which the buyer made an offer. Column~2 restricts attention to the 3,901 matches in which the offer was accepted, and column~3 further restricts to the 3,880 matches in which the product was eventually purchased.

Turning to the results, more aggressive discounts reduce the likelihood of acceptance: a 10 percentage point increase in the discount lowers the acceptance probability by approximately 13.6 percentage points. By contrast, offers accompanied by an immediate payment pledge are 26.8 percentage points more likely to be accepted, holding the discount constant. This acceptance premium is equivalent to an increase of about 19.6\% in the allowable discount while holding the acceptance probability fixed, consistent with Proposition~\ref{prop:commit_premium}. Column~2 shows that buyers who pledge immediate payment are 8.9 percentage points less likely to walk away after acceptance. Finally, column~3 shows that, conditional on eventual sale, pledged buyers complete purchases about 15\% faster $(\exp(-0.162)-1)$. Together, these results indicate that pledges of immediate payment reduce reneging and post-agreement delay and are therefore typically honored in practice. Given these results, I treat counteroffers with pledge of immediate payment as committed offers and those without as non-committed offers.

\input{tables/commitment_walkaway_table}

\section{Identification and Estimation}

My objective is to estimate the following model primitives: the arrival rate of buyers to sellers ($\lambda_S$), the arrival rate of sellers to buyers ($\lambda_B$), the arrival rate of purchase opportunities following a counteroffer ($\lambda_R$), the seller valuation distribution $F_S$, the buyer valuation distribution $F_B$, and the buyer search cost $c$, for a given discount rate $r = 0.05$, numbers of buyers and sellers $N_B$ and $N_S$, and an exogenous post-walkaway trading rate $\kappa$. Identification proceeds in stages. Arrival rates are identified from observed buyer arrivals and post counteroffer purchase delays. Seller and buyer valuation distributions are recovered from optimality and indifference conditions that map observed prices into latent valuations. The search cost is then identified from the indifference between searching and not searching. I describe estimation procedures that directly implement these constructive identification arguments.

%For the discount rate $r$, I report results based on values implied by the daily discrete discount factors used in \textcite{bodoh-creedHowEfficientAre2021}. I convert these into continuous daily discount rates using $\delta = e^{-r}$, yielding rates ranging from $0.02$ to $0.29$. Although these rates may appear high when interpreted as pure time discounting, they are best understood as effective discounting that captures both time preferences and the risk of exit from the market.

The number of buyers $N_B$ and the number of sellers $N_S$ are defined as the number of unique buyers who complete a purchase and the number of listings observed over the sample period, respectively. The exogenous post-walkaway trading rate $\kappa$ is calibrated as the fraction of products that are purchased by another buyer within 1.5 days after a non-committed offer by the original buyer, conditional on no purchase by that buyer, which equals $0.687$.

\paragraph{Arrival rates \(\lambda_S\) and \(\lambda_B\).} I assume that the number of unobservable matches is equal to the number of pure likes, that is, likes not accompanied by a purchase or a comment. The seller arrival rate $\lambda_S$ is then identified as the expected number of distinct buyer arrivals divided by the number of days elapsed from listing to sale, with exposure capped at 30 days to limit the influence of extremely long-lived listings. I set $N_S$ equal to the number of distinct users who list at least one item, and $N_B$ equal to the number of distinct users who complete at least one purchase as a buyer, during the sample period. The buyer arrival rate $\lambda_B$ is then recovered from the steady-state condition \eqref{eq:lambda_b_steady_state}.

\paragraph{Response rate \(\lambda_R\).} The response rate $\lambda_R$ is identified from the time between a committed counteroffer and the subsequent purchase. In the data, I restrict attention to committed counteroffers that do not result in walkaway, and observe for each such match the elapsed time between the counteroffer timestamp and the sale timestamp. In the model, this delay reflects two sequential arrivals governed by the same rate $\lambda_R$, implying an expected delay of $2/\lambda_R$. I therefore identify $\lambda_R$ by equating this mean to the sample mean of the observed delays.

\paragraph{Seller Valuation Distribution \(F_S\).}
\(F_S\) is identified from the optimality condition for committed counteroffer \eqref{eq:committed_offer}, which expresses \(p^1_N(s)\) as a function of \(s\) and the continuation value \(U_S(s)\). The continuation value can be written as the expected discounted realized payoff:
\begin{equation}\label{eq:us_discounted}
U_S(s)
=
\mathbb{E}
\Bigl[
e^{-r T_{\text{sell}}}
\bigl((1-t)p_{\text{sell}} - s\bigr)
\;\big|\; s
\Bigr],
\end{equation}
where \(T_{\text{sell}}\) and \(p_{\text{sell}}\) denote the time to sale and the transaction price, respectively.

Combining \eqref{eq:committed_offer} and \eqref{eq:us_discounted} and solving for \(s\) yields the pseudovalue representation
\begin{equation}\label{eq:pseudovalue}
s
=
\frac{
\mathbb{E}\!\Bigl[
e^{-r T_{\text{sell}}}
(1-t)p_{\text{sell}}
\;\big|\; s
\Bigr]
-
\frac{\lambda_R}{\lambda_R + r}
(1-t)p^{1}_N(s)
}{
\mathbb{E}\!\Bigl[
e^{-r T_{\text{sell}}}
\;\big|\; s
\Bigr]
-
\frac{\lambda_R}{\lambda_R + r}
}.
\end{equation}
By monotonicity of \(p^0(s)\) in \(s\) (Proposition \ref{prop:p0_monotonicity}), conditioning on \(s\) can be replaced by conditioning on the observable \(p^0\) almost surely. Hence the pseudovalue \(s\) is identified for each \(p^0\).

I can thus estimate the conditional expectations in \eqref{eq:pseudovalue} by replacing the conditioning variable $s$ with $p^0 = p^0(s)$. Specifically, I estimate these expectations using local linear regression with a Gaussian kernel.\footnote{The bandwidth is chosen using the normal reference rule, \(h = 1.06 \, \hat\sigma_{p_0} \, n^{-1/5}\), where \(\hat\sigma_{p_0}\) is the sample standard deviation of $p_0$ and $n$ is the sample size. When conditioning on $p_0$ elsewhere, bandwidths are chosen analogously.} I also estimate $p^1_N(s)$ as $\mathbb{E}[p^1_N \mid p^0]$ using the same procedure. Substituting these estimates yields $\hat{s}$ for each observation, and I estimate $F_S$ by the empirical distribution function $\hat{F}_S$ of $\hat{s}$.

\paragraph{Buyer Valuation Distribution \(F_B\).}
The buyer valuation distribution \(F_B\) is identified from the indifference condition between immediate acceptance and making a committed counteroffer (Proposition~\ref{prop:partition}). For a given seller type \(s\), the valuation of the indifferent buyer \(\tau^{AC}(s)\) satisfies
\begin{equation}\label{eq:tauac}
    \tau^{AC}(s)
=
\frac{
p^0(s)
-
\left(\frac{\lambda_R}{\lambda_R + r}\right)^2 p^1_N(s)
}{
1-\left(\frac{\lambda_R}{\lambda_R + r}\right)^2
}.
\end{equation}

The left tail of the buyer valuation distribution at \(\tau^{AC}(s)\) is:
\[
F_B\bigl(\tau^{AC}(s)\bigr)
=
1
-
\mathbb{E}
\Bigl[
\mathbb{I}\!\bigl\{\chi_B(s,b) = A\bigr\}
\;\big|\; s
\Bigr].
\]
The expectation on the right hand side represents, for a given seller \(s\), the probability that an arriving buyer immediately purchases the product (action \(A\)), and this needs to be unconditional on whether the match is observable. See Appendix~\ref{sec:selection} on how I recover this conditional expectation from the knowledge of the seller arrival rate \(\lambda_S\) and the same expectation but conditional on the match being observable.

Because I only observe \(\tau^{AC}(s)\) over a limited range, \(F_B\) is nonparametrically identified only on a restricted support. I therefore parametrize \(F_B\) as a normal distribution with mean \(\mu_B\) and variance \(\sigma_B^2\). I estimate \(\mu_B\) and \(\sigma_B^2\) by least squares:
\[
(\hat\mu_B,\hat\sigma_B^2)
=
\arg\min_{\mu_B,\sigma_B^2}
\sum_{p^0}
\Biggl[
\Bigl(
1
-
\hat{\mathbb{E}}
\bigl[
\mathbb{I}\!\{\chi_B(s,b) = A\}
\;\big|\; p^0
\bigr]
\Bigr)
-
\Phi\!\Bigl(
\hat{\tau}^{AC}(p^0)
\,;\,
\mu_B,\sigma_B^2
\Bigr)
\Biggr]^2
\]
where the summation is over all sellers indexed by their list prices $p^0$, \(\hat{\mathbb{E}}[\cdot \mid p^0]\) denotes the estimated conditional expectation, \(\Phi\) is the standard normal cumulative distribution function, and \(\hat{\tau}^{AC}(p^0)\) is the estimate of \(\tau^{AC}(s)\) for the seller with $p^0 = p^0(s)$, constructed from \eqref{eq:tauac}.

\paragraph{Search Cost $c$.}
Finally, the search cost $c$ is identified from the indifference condition between searching and not searching. We have
\begin{equation}\label{eq:vbcn}
V_B(s, b, CN; p^0) = 
\Bigl(\frac{\lambda_R}{\lambda_R + r}\Bigr)^2
\bigl(
b - p_N^1(s)
\bigr),
\end{equation}
and
\begin{align}\label{eq:vbcs}
V_B(s, b, CS)
&=
\mathbb{E}\Bigl[
e^{-r T_{\text{purchase}}}
\bigl(b - p_{\text{purchase}}\bigr)
-
\int_{0}^{T_{\text{purchase}}}
e^{-rt}\, c \, dt
\;\Big|\; s
\Bigr]\nonumber \\
&=
\underbrace{
\mathbb{E}\Bigl[
e^{-r T_{\text{purchase}}}
\bigl(b - p_{\text{purchase}}\bigr)
\;\Big|\; s, b
\Bigr]
}_{\equiv V_B^1(s,b,CS)}
-
\frac{c}{r}\,
\underbrace{
\mathbb{E}\Bigl[
1 - e^{-r T_{\text{purchase}}}
\;\Big|\; s, b
\Bigr]
}_{\equiv V_B^2(s,b,CS)}.
\end{align}
Here $T_{\text{purchase}}$ and $p_{\text{purchase}}$ denote the time to purchase and the transaction price\footnote{
Equation \eqref{eq:vbcs} integrates $c$ over calendar time, although buyers do not incur $c$ during delays governed by $\lambda_R$ (e.g., after $CN$ or $D$). This leads to a slight overstatement of the time over which search costs are incurred, which is negligible when $\lambda_R$ is sufficiently large (as confirmed in the estimation results).
}.

Given $p_S^1(s,b)$ and $p^0(s)$, the valuation $b$ is recovered from the conditional rank of $p_S^1$. Start with the conditional distribution of non-committed counteroffer:
\begin{equation}\label{eq:p1scdf}
\mathbb{P}(p_S^1 \le p_S^1(s,b)\mid s)
=
\frac{F_b(\tau^{NS}(s)) - F_b(b)}{F_b(\tau^{NS}(s)) - F_b(\tau^{CD}(s))}.
\end{equation}
Using
\[
F_b(\tau^{NS}(s))=\mathbb{E}\!\left[\mathbb{I}\{\chi_B(s,b)\in\{CS,D\}\}\mid s\right],
\quad
F_b(\tau^{CD}(s))=\mathbb{E}\!\left[\mathbb{I}\{\chi_B(s,b)=D\}\mid s\right],
\]
and replacing $s$ with $p^0=p^0(s)$, this yields the plug-in estimator for the pseudo-value $b$:
\[
\hat b(p^0,p_S^1)
=
\hat F_b^{-1}\!\left(
\hat{\mathbb E}\!\left[
\mathbb{I}\{\chi_B\in\{CS,D\}\}
\mid p^0
\right]
-
\mathbb{P}(p_S^{1\prime}\le p_S^1\mid p^0)
\hat{\mathbb E}\!\left[
\mathbb{I}\{\chi_B=CS\}
\mid p^0
\right]
\right).
\]

Substituting $\hat b(p^0,p_S^1)$ into \eqref{eq:vbcn}--\eqref{eq:vbcs} and replacing conditioning variables $(s,b)$ with $(p^0,p_S^1)$ yields estimators $\hat V_B(p^0,p_S^1,CN)$, $\hat V_B^1(p^0,p_S^1,CS)$, and $\hat V_B^2(p^0,p_S^1,CS)$. These conditional expectations are estimated using bivariate local linear regression with a Gaussian kernel.

Evaluating the lower tail $\underline{p}_S^1(p^0)$ of the conditional distribution of $p_S^1$ given $p^0$, defined as the 0.1st percentile, we obtain
\[
\hat V_B(p^0,\underline{p}_S^1(p^0),CN)
\approx
\hat V_B^1(p^0,\underline{p}_S^1(p^0),CS)
-
\frac{c}{r}\,
\hat V_B^2(p^0,\underline{p}_S^1(p^0),CS),
\]
where the approximation follows because the marginal buyer with $b=\tau^{NS}(s)$ corresponds to the lowest non-committed offer by monotonicity (Proposition \ref{prop:p1s_monotonicity}). Finally, I integrate this expression over $p^0$ using its empirical distribution and solve for $c$ to obtain the estimate $\hat{c}$.

\section{Results}

Estimation results are reported in Table \ref{tab:estimation_bootstrap_se_single_r} (Table \ref{tab:estimation_summary_panelc_only} shows how estimates vary with other values of the discount rate $r$). Bootstrap standard errors are reported in parentheses. Panel A reports parameters that are not estimated, either calibrated or directly observed, while Panel B reports estimates for parameters that are invariant across specifications. The estimated response rate is \( \lambda_R = 3.40 \), implying that a response or purchase opportunity arrives on average once every 7.06 hours. The buyer arrival rate to a seller is estimated to be \( \lambda_S = 0.64 \), meaning that a seller encounters a buyer approximately once every 1.56 days, which is substantially slower than the response rate. Note that the estimates are consistent with the assumption that $\lambda_R > \lambda_S$. Conversely, the seller arrival rate to a buyer is \( \lambda_B = 0.80 \), reflecting the fact that sellers are relatively more abundant than buyers in the market, so buyers encounter sellers more frequently than sellers encounter buyers.

Turning to Panel~C, which reports estimates of the search cost and valuation distributions, the implied search cost is economically nonnegligible. The estimated flow search cost is 819.99 yen per day. While the bootstrap standard error is large (441.59), this largely reflects the right-skewed nature of the bootstrap distribution, which arises from the ratio-of-means structure of the estimator. The actual 95\% bootstrap percentile confidence interval, [54.32, 1727.35], excludes zero, indicating that the estimate is statistically distinguishable from zero. The remaining columns report quartiles of the seller and buyer valuation distributions. The point estimate of the first quartile of buyer valuations is negative, but the lower quantiles, including the median, are imprecisely estimated. Buyers in the upper quartile exhibit positive gains from trade even when matched with high valuation sellers. Nevertheless, trade does not always occur, which can be explained by agents' continuation values from matching with alternative partners, even in this complete information environment.

In terms of model fit, although list prices are not targeted moments in the estimation, the model reproduces their empirical distribution reasonably well at this value. Simulating the equilibrium using the estimated parameters (see Online Appendix~\ref{sec:vfi} for details) yields a mean list price of 14{,}493.6 yen in the data versus 14{,}492.5 yen in the model and a median of 14{,}145.8 yen in the data versus 14{,}572.9 yen in the model. The model understates dispersion, with a standard deviation of 2{,}817.5 yen in the data compared with 1{,}465.0 yen in the model, and compresses the upper tail of the distribution.

\input{tables/table_estimation_bootstrap_se_single_r}

\section{Counterfactual Analysis}
Using the estimated model, I conduct a counterfactual experiment in which buyers are fully committed to purchasing once their offers are accepted. I implement this scenario by imposing an effectively infinite search cost during the bargaining stage, thereby eliminating post acceptance search. Under this full commitment policy, buyers must complete the purchase after an offer is accepted and cannot search for other sellers. The search cost faced by unmatched buyers is held fixed, so only search during the bargaining stage is shut down. I compute the stationary equilibrium by value function iteration on \(U_S\), \(U_B\), and \(V_B\). Details of the simulation procedure are provided in Online Appendix \ref{sec:vfi}.

I measure the welfare of a seller with valuation \(s\) by \(U_S(s)\) and the welfare of a buyer with valuation \(b\) by \(U_B(b)\). Platform welfare is measured by the platform's continuation value per listing, denoted by \(U_P\). This value satisfies the recursive equation
\begin{equation}\label{eq:platform_welfare}
\begin{aligned}
U_P
=
\int\!\!\int
\Biggl[
&t\,\mathds{1}\!\left\{\chi_B(s,b)=A\right\}\, p^0(s)
\\
&+
t\,\mathds{1}\!\left\{\chi_B(s,b)=CN\right\}
\left(\frac{\lambda_R}{\lambda_R + r}\right)^2 p^1_N(s)
\\
&+
t\,\mathds{1}\!\left\{\chi_B(s,b)=CS\right\}
\frac{\lambda_B}{\lambda_B + r}
\Bigl(1-(1-\kappa)\,F_S\!\bigl(s^*(b,p^1_S(s,b))\bigr)\Bigr)
p^1_S(s,b)
\\
&+
\Bigl[
\mathds{1}\!\left\{\chi_B(s,b)=D\right\}
+
\mathds{1}\!\left\{\chi_B(s,b)=CS\right\}
(1-\kappa)\,F_S\!\bigl(s^*(b,p^1_S(s,b))\bigr)
\Bigr] U_P
\Biggr]
\, dF_S(s)\, dF_B(b).
\end{aligned}
\end{equation}
The first three terms are the platform's expected discounted commission revenues from the current match under acceptance, committed counteroffer, and non-committed counteroffer. The final term captures continuation: when the buyer declines immediately or walks away after a non-committed counteroffer, the listing remains active and the platform retains continuation value \(U_P\).

\noindent Overall welfare is measured on a per listing basis as the sum of expected seller welfare, buyer welfare scaled by the buyer to seller ratio, and platform welfare per listing:
\begin{equation}\label{eq:overall_welfare}
\begin{aligned}
W
&=
\int U_S(s)\, dF_S(s)
+
\frac{N_B}{N_S}\int U_B(b)\, dF_B(b)
+
U_P.
\end{aligned}
\end{equation}

Table~\ref{tab:welfare_unified_r0p05} reports the aggregate welfare effects of introducing a full commitment policy based on the aforementioned welfare measures. Overall, total welfare increases modestly, with a gain of 137.65 yen per listing. Sellers gain on average 266.78 yen per listing, while buyers lose 194.72 yen on average. Platform welfare also increases by 27.67 yen per listing. The table further reports heterogeneous welfare effects by valuation quartiles for sellers and buyers, and the magnitude of these impacts varies substantially across types. Seller welfare effects are positive throughout the distribution and increase from the lower to the middle of the distribution: gains are 221.29 yen in the first quartile and 248.88 yen in the second quartile, rising to 379.62 yen in the third quartile, before moderating to 341.27 yen in the top quartile. Buyer welfare effects, by contrast, are uniformly negative across the distribution, with losses increasing in buyer valuation. Buyers in the bottom quartile incur a loss of 61.11 yen, compared with losses of 181.22 yen, 293.92 yen, and 325.02 yen in the second, third, and fourth quartiles, respectively.

\input{tables/welfare_unified_r0p05}

To understand the internal mechanism behind these results, Table~\ref{tab:action_p0_p1_baseline_diff_r0p05} reports the fractions of each buyer action $\chi_B(s,b)\in{A,C,D}$, along with list prices $p^0$ and accepted counteroffers $p^1$, both overall and by quartiles of the seller valuation $s$, as well as their changes under the counterfactual. Overall, the introduction of full commitment reduces counteroffers $(\Delta C=-0.092)$ and increases outright declines $(\Delta D=0.073)$, while slightly increasing immediate purchases $(\Delta A=0.019)$. These shifts are accompanied by an increase in list prices $(\Delta p^0=281.4)$ and a decline in accepted counteroffer prices $(\Delta p^1=-353.5)$.

The effects are strongly heterogeneous across seller types. Sellers in the top quartile experience a sharp reduction in counteroffers \((\Delta C = -0.188)\), with most buyers switching to outright declines \((\Delta D = 0.171)\). Because these sellers no longer face delayed purchases and walkaways, their continuation values increase, lowering the opportunity cost of raising list prices and being declined, and they respond by raising list prices substantially \((\Delta p^0 = 442.2)\). Sellers with lower valuations are also meaningfully affected. In the second quartile, counteroffers decline \((\Delta C = -0.089)\) and immediate purchases increase \((\Delta A = 0.059)\), while in the third quartile immediate purchases decrease slightly \((\Delta A = -0.004)\) alongside a reduction in counteroffers \((\Delta C = -0.068)\). Across all quartiles, list prices increase by 160.8, 60.3, 462.3, and 442.2 yen in the first through fourth quartiles, respectively. This upward shift in list prices improves continuation values and raises the level of committed counteroffers directed toward lower-valuation sellers.

The welfare increase is driven primarily by seller gains, while buyers experience welfare losses and the platform experiences a modest welfare gain. When matched with higher-valuation sellers, buyers face substantially higher list prices, which is particularly detrimental for higher-valuation buyers, who are more likely to transact at list prices. In contrast, low-valuation buyers are less affected, as their outcomes are largely determined by whether they match with lower-valuation sellers. For the platform, full commitment reduces negotiated trade, reflected in fewer counteroffers and more declines $(\Delta C=-0.092$ and $\Delta D=0.073)$, alongside a reallocation toward immediate purchases $(\Delta A=0.019)$, and lowers negotiated prices $(\Delta p^{1}=-353.5)$. At the same time, however, list prices increase substantially $(\Delta p^{0}=281.4)$, so the platform earns more commission revenue on transactions that occur at the posted price. This gain from higher list-price transactions outweighs the revenue loss from fewer counteroffers and lower negotiated prices, yielding a modest increase in platform welfare overall. Overall, a full commitment policy benefits sellers and modestly benefits the platform, but harms buyers.

\input{tables/action_p0_p1_baseline_diff_by_s_quartile_r0p05}

\section{Conclusion}

This paper is the first to study price negotiations on online marketplaces while explicitly accounting for the fact that buyers naturally search for alternative sellers, which creates a commitment problem in completing payments, a feature that has long been pervasive in practice. I build a structural model of bargaining with buyer search in which buyers self-select into commitment status, and sellers require a premium from non-committed buyers when accepting their counteroffers. Using detailed textual data from a Japanese online platform, I detect buyer attempts to credibly convey immediate payment and show that sellers accept lower counteroffers from such buyers, consistent with the model's predictions.

Counterfactual analysis shows that a full commitment policy increases overall welfare, driven primarily by seller gains that outweigh buyer losses and are accompanied by a modest gain for the platform. In equilibrium, sellers, particularly those with higher valuations, respond by raising their list prices, while the elimination of the option value of search during bargaining further reduces buyer welfare. Although the platform experiences a reduction in negotiated trade, both through fewer counteroffers and more declines at the extensive margin and through lower negotiated prices at the intensive margin, the increase in list prices more than offsets these losses and yields a modest increase in platform welfare. Although these quantitative results are specific to the institutional setting and estimated primitives in this market, the underlying framework and economic mechanism apply more broadly to bargaining environments in which buyers can continue searching while negotiating.

While platforms such as eBay have implemented randomized rollouts of commitment rules, such experiments capture only a limited range of behavior and cannot evaluate platform-wide equilibrium effects. This analysis complements those efforts by providing a model-based framework for studying platform design choices. At the same time, the framework is simplified and relies on strong assumptions, most notably complete information, which are imposed for tractability and to enable counterfactual analysis. Future work could relax these assumptions by estimating models with incomplete information or richer inter bargaining dynamics, following recent work on static models (see \textcite{freyberger2025well}; \textcite{larsen2021efficiency}), though possibly at the cost of a tractable equilibrium.

\printbibliography

\appendix
\renewcommand{\thesection}{\Alph{section}}

\section{Appendix}

\subsection{Regularity conditions}\label{sec:regularity}

\begin{assumption}[Regularity I]\label{ass:regularity_wow}
For all $(s,b)$,
\begin{equation}\label{eq:reg1}
    \Bigl(1 - F_S\!\bigl(s^*(b,p^1_S(s,b))\bigr)\Bigr)
    \frac{H_b\!\bigl(b,p^1_S(s,b)\bigr)}{H_{p^1}\!\bigl(b,p^1_S(s,b)\bigr)}
    <
    \frac{\delta_R^2}{\delta_B}-1 ,
\end{equation}
where
\[
H_b(b,p^1)
=
\delta_B\,
(1-\kappa)\,
f_S\!\bigl(s^*(b,p^1)\bigr)\,
\frac{\partial s^*(b,p^1)}{\partial b}\,
\Bigl(U_S(s)-(1-t)p^1+s\Bigr),
\]
and
\begin{equation}\label{eq:hp1}
\begin{aligned}
H_{p^1}(b,p^1)
=
\delta_B
\Bigl[
&(1-t)\Bigl(1-(1-\kappa)F_S\!\bigl(s^*(b,p^1)\bigr)\Bigr)
\\
&+
(1-\kappa)\,
f_S\!\bigl(s^*(b,p^1)\bigr)\,
\frac{\partial s^*(b,p^1)}{\partial p^1}\,
\Bigl(U_S(s)-(1-t)p^1+s\Bigr)
\Bigr].
\end{aligned}
\end{equation}
Here
\[
\delta_R=\frac{\lambda_R}{\lambda_R+r},
\qquad
\delta_B=\frac{\lambda_B}{\lambda_B+r},
\]
and $s^*(b,p^1)$ is defined by the cutoff condition
$V_B\!\bigl(s^*(b,p^1),b\bigr)=b-p^1$.
\end{assumption}
\noindent The ratio ${H_b} / {H_{p^1}}$ captures the negative of the sensitivity of the non-committed counteroffer $p^1_S(s,b)$ with respect to the buyer valuation $b$, as implied by the implicit function theorem applied to the seller acceptance condition \eqref{eq:seller_accept_noncommit}. The left-hand side of \eqref{eq:reg1} can therefore be interpreted as the discounted marginal benefit, from a marginal increase in $b$, of having a lower non-committed offer accepted and not walking away. The condition requires that this marginal benefit not exceed the right-hand side of \eqref{eq:reg1}, which reflects the heavier discounting faced by non-committed buyers.

\begin{assumption}[Regularity II]\label{ass:regularity_sc}
\[
H_{p^1}\!\bigl(b,p^1_S(s,b)\bigr) > 0,
\]
where $H_{p^1}(b,p^1)$ is defined in \eqref{eq:hp1}.
\end{assumption}
\noindent
The first component of the derivative \eqref{eq:hp1} is positive, reflecting the higher surplus conditional on trade as the counteroffer rises.
The second component is negative, reflecting the increased probability that the buyer walks away as the counteroffer rises.
Assumption~\ref{ass:regularity_sc} requires that the former effect dominates the latter.

\subsection{Proofs for Propositions}

\begin{proof}[Proof for Proposition \ref{prop:p1s_existence}]\label{proof:p1s_existence}
Fix $s$ and $b$. Under the maintained assumption in Proposition~\ref{prop:p1s_existence}, $U_S(s)>0$. Define $l(p^1)\equiv (1-t)p^1-s$ and $\delta\equiv \lambda_B/(\lambda_B+r)\in(0,1)$. Let
\[
w(p^1)\equiv (1-\kappa)F_S\!\bigl(s^*(b,p^1)\bigr),
\qquad
G(p^1)\equiv
\delta\Bigl(w(p^1)U_S(s)+\bigl[1-w(p^1)\bigr]l(p^1)\Bigr)-U_S(s).
\]
By Assumption~\ref{ass:monotonicity}, for each $b$ the function $V_B(\cdot,b)$ is continuous and strictly decreasing in $s$. Hence for any $(b,p^1)$ the cutoff $s^*(b,p^1)$ defined by $V_B(s^*(b,p^1),b)=b-p^1$ is well defined, and the mapping $p^1\mapsto s^*(b,p^1)$ is continuous. Since $F_S$ is absolutely continuous, it is continuous, so $w(\cdot)$ is continuous. Because $l(\cdot)$ is linear, $G(\cdot)$ is continuous.

Let $\underline{p}^1\equiv (U_S(s)+s)/(1-t)$ so that $l(\underline{p}^1)=U_S(s)$. Then
\[
G(\underline{p}^1)
=
\delta\Bigl(w(\underline{p}^1)U_S(s)+\bigl[1-w(\underline{p}^1)\bigr]U_S(s)\Bigr)-U_S(s)
=
(\delta-1)U_S(s)
<
0.
\]
Next, since $0\le w(p^1)\le 1-\kappa$, we have $1-w(p^1)\ge \kappa>0$ for all $p^1$. Using $U_S(s)>0$ and $w(p^1)\ge 0$,
\[
G(p^1)
=
\delta\Bigl(w(p^1)U_S(s)+\bigl[1-w(p^1)\bigr]l(p^1)\Bigr)-U_S(s)
\ge
\delta\,\kappa\,l(p^1)-U_S(s).
\]
Because $l(p^1)=(1-t)p^1-s$ diverges to $+\infty$ as $p^1\to\infty$, the right hand side diverges to $+\infty$. Hence there exists $\overline{p}^1$ such that $G(\overline{p}^1)>0$.
By continuity of $G$, the intermediate value theorem implies that there exists $p^{1}_S\in(\underline{p}^1,\overline{p}^1)$ with $G(p^{1}_S)=0$, that is, \eqref{eq:seller_accept_noncommit} holds with equality at $p^{1}_S$.
\end{proof}

\begin{proof}[Proof for Proposition \ref{prop:commit_premium}]\label{proof:commit_premium}
Suppose that $p^1_S(s,b)$ exists in a neighborhood of $b$. Then the seller acceptance condition
\eqref{eq:seller_accept_noncommit} holds with equality. Likewise, $p^1_N(s)$ satisfies
\eqref{eq:seller_accept_commit} with equality.
Since $\frac{\lambda_B}{\lambda_B + r}< \frac{\lambda_R}{\lambda_R+r}$, combining the two equalities implies
\begin{equation}\label{eq:ps_ge_pn_step}
\begin{aligned}
&(1 - \kappa) F_S\!\bigl(s^*(b,p^1_S(s,b))\bigr)\,U_S(s)
\\
&\qquad
+\Bigl[1-(1 - \kappa) F_S\!\bigl(s^*(b,p^1_S(s,b))\bigr)\Bigr]
\bigl((1-t)p^1_S(s,b)-s\bigr)
\;>\;
(1-t)p^1_N(s)-s .
\end{aligned}
\end{equation}
The left-hand side of \eqref{eq:ps_ge_pn_step} is a convex combination of $U_S(s)$ and
$(1-t)p^1_S(s,b)-s$, with weight $(1 - \kappa) F_S\!\bigl(s^*(b,p^1_S(s,b))\bigr)\in[0,1]$ on the former.
Since $U_S(s)\le (1-t)p^1_S(s,b)-s$ at the acceptance threshold, the expression is bounded above by
$(1-t)p^1_S(s,b)-s$.
Therefore,
\[
(1-t)p^1_S(s,b)-s \;>\; (1-t)p^1_N(s)-s,
\]
which implies
\[
p^1_S(s,b)\;>\;p^1_N(s).
\]
\end{proof}

\begin{proof}[Proof of Proposition~\ref{prop:partition}]\label{proof:partition}
Fix $s$.  
Define $\delta_R=\frac{\lambda_R}{\lambda_R+r}$ and $\delta_B=\frac{\lambda_B}{\lambda_B+r}$.
Let $s^*(b,p^1)$ be defined by $V_B\!\bigl(s^*(b,p^1),b\bigr)=b-p^1$.
Define the seller acceptance function
\[
H(b,p^1)
:=
\delta_B\Bigl(
(1-\kappa)F_S\!\bigl(s^*(b,p^1)\bigr)\,U_S(s)
+\bigl[1-(1-\kappa)F_S\!\bigl(s^*(b,p^1)\bigr)\bigr]\bigl((1-t)p^1-s\bigr)
\Bigr)
-U_S(s),
\]
so that the non-committed counteroffer is characterized by
$H\bigl(b,p^1_S(s,b)\bigr)=0$ (see \eqref{eq:seller_accept_noncommit}).

Differentiating $H$ yields
\[
H_b(b,p^1)
=
\delta_B(1-\kappa)\,
f_S\!\bigl(s^*(b,p^1)\bigr)\,
\frac{\partial s^*(b,p^1)}{\partial b}\,
\Bigl(U_S(s)-(1-t)p^1+s\Bigr),
\]
and
\[
\begin{aligned}
H_{p^1}(b,p^1)
=
\delta_B\Bigl[
&(1-t)\Bigl(1-(1-\kappa)F_S\!\bigl(s^*(b,p^1)\bigr)\Bigr)
\\
&+
(1-\kappa)\,
f_S\!\bigl(s^*(b,p^1)\bigr)\,
\frac{\partial s^*(b,p^1)}{\partial p^1}\,
\Bigl(U_S(s)-(1-t)p^1+s\Bigr)
\Bigr].
\end{aligned}
\]
Whenever $H_{p^1}\neq 0$, the implicit function theorem implies
\[
\frac{\partial p^1_S(s,b)}{\partial b}
=
-\,\frac{H_b\!\bigl(b,p^1_S(s,b)\bigr)}{H_{p^1}\!\bigl(b,p^1_S(s,b)\bigr)}.
\]
Assumption~\ref{ass:regularity_wow} implies %$H_{p^1}\!\bigl(b,p^1_S(s,b)\bigr)>0$ and hence
\begin{equation}\label{eq:reg_used_partition}
\Bigl(1-F_S\!\bigl(s^*(b,p^1_S(s,b))\bigr)\Bigr)
\Bigl(-\frac{\partial p^1_S(s,b)}{\partial b}\Bigr)
%=
%\Bigl(1-F_S\!\bigl(s^*(b,p^1_S(s,b))\bigr)\Bigr)
%\frac{H_b\!\bigl(b,p^1_S(s,b)\bigr)}{H_{p^1}\!\bigl(b,p^1_S(s,b)\bigr)}
<
\frac{\delta_R^2}{\delta_B}-1 .
\end{equation}

For immediate purchase,
\[
\frac{\partial V_B(s,b;A,p^0)}{\partial b}=1.
\]
For a committed counteroffer,
\[
\frac{\partial V_B(s,b;CN,p^0)}{\partial b}
=
\delta_R^2.
\]
For decline,
\[
\frac{\partial V_B(s,b;D,p^0)}{\partial b}
=
\delta_R\,\delta_B
\int
\frac{\partial V_B(s',b)}{\partial b}
\,dF_S(s')
\;\le\;
\delta_R\,\delta_B .
\]
where the inequality follows from the Lipschitz--$1$ condition in Assumption~\ref{ass:monotonicity}.

For a non-committed counteroffer, applying Leibniz' rule yields:
\[
\frac{\partial V_B(s,b;CS,p^0)}{\partial b}
=
\delta_B
\Biggl[
\int_{-\infty}^{s^*(b,p^1_S)}
\frac{\partial V_B(s',b)}{\partial b}\,dF_S(s')
+
\bigl(1-F_S(s^*(b,p^1_S))\bigr)
\Bigl(1-\frac{\partial p^1_S(s,b)}{\partial b}\Bigr)
\Biggr].
\]
By the Lipschitz--$1$ condition, $\frac{\partial V_B(s',b)}{\partial b}\le 1$ whenever the derivative exists, hence
\[
\int_{-\infty}^{s^*(b,p^1_S)}\frac{\partial V_B(s',b)}{\partial b}\,dF_S(s')
\le
F_S\!\bigl(s^*(b,p^1_S)\bigr).
\]
Therefore,
\[
\begin{aligned}
\frac{\partial V_B(s,b;CS,p^0)}{\partial b}
&\le
\delta_B\Bigl[
F_S(s^*)+(1-F_S(s^*))\Bigl(1-\frac{\partial p^1_S}{\partial b}\Bigr)
\Bigr]
\\
&=
\delta_B\Bigl[
1+(1-F_S(s^*))\Bigl(-\frac{\partial p^1_S}{\partial b}\Bigr)
\Bigr]
\\
&<
\delta_B\Bigl[1+\Bigl(\frac{\delta_R^2}{\delta_B}-1\Bigr)\Bigr]
=
\delta_R^2,
\end{aligned}
\]
where the strict inequality uses \eqref{eq:reg_used_partition}.

Given $\lambda_B<\lambda_R$, the marginal values satisfy
\[
\frac{\partial V_B(s,b;D)}{\partial b}
<
\frac{\partial V_B(s,b;CS)}{\partial b}
<
\frac{\partial V_B(s,b;CN)}{\partial b}
<
\frac{\partial V_B(s,b;A)}{\partial b}
\]
It follows that each of the differences
\[
V_B(s,b;CS)-V_B(s,b;D),\quad
V_B(s,b;CN)-V_B(s,b;CS),\quad
V_B(s,b;A)-V_B(s,b;CN)
\]
is weakly increasing in $b$, so each adjacent pair of value functions crosses at most once.
Since $\mathcal{M}(s)$ is nonempty, this single-crossing property implies that the buyer's optimal action admits a unique monotone partition of the type space, ordered as
$D \prec CS \prec CN \prec A$.

\end{proof}

\begin{proof}[Proof of Proposition~\ref{prop:p0_monotonicity}]\label{proof:p0_monotonicity}
Define
\[
\delta_R:=\frac{\lambda_R}{\lambda_R+r}
\qquad\text{and}\qquad
\delta_S:=\frac{\lambda_S}{\lambda_S+r}.
\]

Since \(\{A,CN,D\}\subset M(s)\) for all \(s\in\mathcal S\), Proposition~\ref{prop:partition} implies that buyer behavior is characterized by cutoff rules. In particular, there is a cutoff \(\tau^{AC}(p^0,s)\) such that buyers choose \(A\) if and only if \(b\ge \tau^{AC}(p^0,s)\), and a cutoff \(\tau^{CD}(s)\) such that buyers choose \(D\) rather than \(CN\) if and only if \(b<\tau^{CD}(s)\).

The seller's objective can therefore be written as
\[
\delta_S\Bigl[
\bigl(1-F_B(\tau^{AC}(p^0,s))\bigr)\bigl((1-t)p^0-s\bigr)
+\bigl(F_B(\tau^{AC}(p^0,s))-F_B(\tau^{CD}(s))\bigr)\delta_R U_S(s)
+F_B(\tau^{CD}(s))U_S(s)
\Bigr].
\]
Since \(\delta_S>0\) and \(\tau^{CD}(s)\) does not depend on \(p^0\), the seller chooses \(p^0\) to maximize
\begin{equation}\label{eq:seller_obj}
\Psi(p^0,s)
=
\bigl(1-F_B(\tau^{AC}(p^0,s))\bigr)\bigl((1-t)p^0-s\bigr)
+
F_B(\tau^{AC}(p^0,s))\,\delta_R U_S(s).
\end{equation}

I now reparametrize the seller's problem in terms of the cutoff \(\tau=\tau^{AC}(p^0,s)\). At the \(A\)--\(CN\) margin, the buyer is indifferent between accepting immediately and making a committed counteroffer, so
\[
\tau^{AC}(p^0,s)-p^0
=
\delta_R^2\bigl(\tau^{AC}(p^0,s)-p_N^1(s)\bigr).
\]
Solving for \(p^0\) yields
\begin{equation}\label{eq:p0_p1_relation}
p^0
=
(1-\delta_R^2)\tau^{AC}(p^0,s)+\delta_R^2 p_N^1(s).
\end{equation}
Moreover, from \eqref{eq:committed_offer}, the committed counteroffer satisfies
\begin{equation}\label{eq:committed_offer_delta}
    p_N^1(s)=\frac{1}{1-t}\left(s+\frac{1}{\delta_R}U_S(s)\right).
\end{equation}
Substituting \eqref{eq:p0_p1_relation} and \eqref{eq:committed_offer_delta} into \eqref{eq:seller_obj}, we obtain
\[
\Psi(p^0,s)
=
\delta_R U_S(s)
+
(1-\delta_R^2)\bigl(1-F_B(\tau^{AC}(p^0,s))\bigr)\bigl((1-t)\tau^{AC}(p^0,s)-s\bigr).
\]
Since the first term does not depend on \(p^0\) and \(1-\delta_R^2>0\), maximizing over \(p^0\) is equivalent to maximizing
\[
\tilde{\Psi}(\tau,s)
=
\bigl(1-F_B(\tau)\bigr)\bigl((1-t)\tau-s\bigr)
\]
over \(\tau\).

The first-order condition is
\[
(1-t)\bigl(1-F_B(\tau)\bigr)-f_B(\tau)\bigl((1-t)\tau-s\bigr)=0,
\]
which can be rewritten as
\[
\tau-\frac{1-F_B(\tau)}{f_B(\tau)}=\frac{s}{1-t}.
\]
Because \(F_B\) has increasing hazard rate, the inverse hazard ratio \((1-F_B(\tau))/f_B(\tau)\) is weakly decreasing, and hence left hand side is strictly increasing. Therefore the first-order condition has a unique solution, denoted \(\tau^*(s)\). Since the right-hand side is strictly increasing in \(s\), it follows that \(\tau^*(s)\) is strictly increasing in \(s\).

It remains to show that \(p^0(s)\) is strictly increasing in \(s\). By \eqref{eq:p0_p1_relation},
\[
p^0(s)=(1-\delta_R^2)\tau^*(s)+\delta_R^2 p_N^1(s),
\]
so it is enough to show that \(p_N^1(s)\) is weakly increasing.
Under Assumption~\ref{ass:monotonicity}, \(V_S(\cdot,b)\) is Lipschitz with constant one, which implies that \(U_S\) is \(\delta_S\)-Lipschitz. Hence, for any \(s_2>s_1\),
\[
U_S(s_2)-U_S(s_1)\ge -\delta_S(s_2-s_1).
\]
If \(\lambda_R\ge \lambda_S\), then \(\delta_R\ge \delta_S\), and therefore, from \eqref{eq:committed_offer_delta}, 
\[
p_N^1(s_2)-p_N^1(s_1)
=
\frac{1}{1-t}\left[(s_2-s_1)+\frac{U_S(s_2)-U_S(s_1)}{\delta_R}\right]
\ge
\frac{1-\delta_S/\delta_R}{1-t}(s_2-s_1)\ge 0.
\]
Thus \(p_N^1(s)\) is weakly increasing in \(s\). Since \(\tau^*(s)\) is strictly increasing and \(1-\delta_R^2>0\), it follows that \(p^0(s)\) is strictly increasing in \(s\). Uniqueness of \(p^0(s)\) follows from uniqueness of \(\tau^*(s)\) and the affine relation \eqref{eq:p0_p1_relation}.
\end{proof}
\begin{proof}[Proof of Proposition~\ref{prop:p1s_monotonicity}]\label{proof:p1s_monotonicity}
Since $p^1_S(s,b)$ solves the acceptance condition \eqref{eq:seller_accept_noncommit} with equality, the implicit function theorem implies
\[
\frac{\partial p^1_S}{\partial b}(s,b)
=
-\frac{H_b\!\bigl(b,p^1_S(s,b)\bigr)}{H_{p^1}\!\bigl(b,p^1_S(s,b)\bigr)}.
\]
Assumption~\ref{ass:regularity_sc} gives $H_{p^1}\!\bigl(b,p^1_S(s,b)\bigr)>0$, so it suffices to show
$H_b\!\bigl(b,p^1_S(s,b)\bigr)>0$.

By the definition of $H_b$ in Assumption~\ref{ass:regularity_wow},
\[
H_b(b,p^1)
=
\delta_B(1-\kappa)\,
f_S\!\bigl(s^*(b,p^1)\bigr)\,
\frac{\partial s^*(b,p^1)}{\partial b}\,
\Bigl(U_S(s)-(1-t)p^1+s\Bigr).
\]
The first three objects are strictly positive provided that $F_S$ has positive density, and the last are strictly negative evauated at $p^1 = p^1_S(s, b)$.
First, $\frac{\partial s^*(b,p^1)}{\partial b} < 0$ follows from Assumption~\ref{ass:monotonicity}:
since $V_B(s,b)$ is decreasing in $s$ and Lipschitz in $b$ with constant one, the cutoff defined by
$V_B(s^*(b,p^1),b)=b-p^1$ is weakly decreasing in $b$.
\eqref{eq:vbchi} further implies that as long as $\chi_B(s, b) = CS$, $\frac{\partial V_B(s, b)}{\partial b} < 1$, hence $s^*(b,p^1)$ is strictly decreasing in $b$.
Second, by \eqref{prop:commit_premium}, $(1-t)p^1_S - s > (1-t)p^1_N - s = \frac{\lambda_R + r}{\lambda_R} U_S(s) > U(s)$, hence $U_S(s)-(1-t)p^1+s < 0$.
Therefore $H_b\!\bigl(b,p^1_S(s,b)\bigr) > 0$.
\end{proof}

\subsection{Hedonic regression results}\label{sec:hedonic_regression}

This appendix reports the hedonic regression used to residualize list prices and offers. The specification controls for observable item characteristics and shipping attributes. The fitted values are used to construct residualized prices and offers that enter the main analysis.
\input{tables/table_hedonic}

\subsection{Estimation Details}

\subsubsection{Recovering unconditional choice probabilities.}\label{sec:selection}
I denote $\tilde{\lambda}_S(s)$ as the arrival rate of an observable match to seller $s$.
I estimate $\tilde{\lambda}_S(s)$ based on the days $T^{\mathrm{obs}}$ between consecutive observed arrivals for a given listing. Then $\tilde{\lambda}_S(s)$ is identified:
\[
\tilde{\lambda}_S(s) = \frac{1}{\mathbb{E}[T^{obs}|s]}
\]
To avoid numerical instability due to inversion, I first estimate $\mathbb{E}[\log T^{\mathrm{obs}}\mid s]$ after dropping very short intervals (less than ten minutes), and apply Duan's smearing correction to recover an estimate of $\mathbb{E}[T^{\mathrm{obs}}\mid s]$.
I then define $q(s)$ as the probability a match is observed.
Then we have:
\[
\tilde{\lambda}_S(s) = q(s) \lambda_S
\]
I recover $q(s)$ from this relationship.

I denote $p_{\chi}(s)=\mathbb{P}(\chi_B(s,b)=\chi \mid s)$ for $\chi\in\{A,CN,CS,D\}$. Let $\tilde{p}_{\chi}(s)$ be the probability that action $\chi$ is taken conditional on an observable match. Then, because $\chi\in\{A,CN,CS\}$ is always observable,
\[
p_{\chi}(s)=q(s)\tilde{p}_{\chi}(s),
\qquad \chi\in\{A,CN,CS\}.
\]
Finally, the unconditional decline probability is obtained as the residual
\[
p_D(s)=1-p_A(s)-p_{CN}(s)-p_{CS}(s).
\]

\clearpage

\subsubsection{Sensitivity Analysis}\label{sec:sensitivity}
\input{tables/table_estimation_summary_panelc_only}

\clearpage

\section*{Online Appendix (Not for Print)}
\addcontentsline{toc}{section}{Online Appendix}

\input{draft/online_appendix}

\end{document}

%% file: figures/sequence_tree.tex
\begin{figure}[htbp]
\centering
\begin{forest}
for tree={
  grow'=south,
  parent anchor=south,
  child anchor=north,
  edge={->},
  align=center,
  font=\small,
  inner sep=2pt,
  s sep=10mm,
  l sep=8mm,
}
[
  {$\text{Seller } s$\\ $p^0$}
  [
    {$\text{Buyer } b$}
    [
      {$A$}
      [{$(b-p^0,\ (1-t)p^0-s)$}]
    ]
    [
      {$D$}
      [{$(U_B(b),\ U_S(s))$}]
    ]
    [
      {$C:\ p^1$}
      [
        {$\text{Buyer } b$}
        [
          {$S$}
          [
            {$\text{Seller } s$}
            [
              {$D$}
              [{$(U_B(b),\ U_S(s))$}]
            ]
            [
              {$A$}
              [
                {$\text{Buyer } b$}
                [
                  {$W$}
                  [{$\bigl(V_B(s',b),\ \kappa((1-t)p^1-s)+(1-\kappa)U_S(s)\bigr)$}]
                ]
                [
                  {$P$}
                  [{$(b-p^1,\ (1-t)p^1-s)$}]
                ]
              ]
            ]
          ]
        ]
        [
          {$N$}
          [
            {$\text{Seller } s$}
            [
              {$D$}
              [{$(U_B(b),\ U_S(s))$}]
            ]
            [
              {$A$}
              [{$(b-p^1,\ (1-t)p^1-s)$}]
            ]
          ]
        ]
      ]
    ]
  ]
]
\end{forest}
\caption{Sequence of events}
\label{fig:sequence_tree}

\begin{flushleft}
\footnotesize
\textit{Note:} Decision nodes are labeled by the acting player. Terminal nodes report the buyer payoff first and the seller payoff second, both undiscounted. Stochastic arrival processes and valuation draws by Nature are suppressed.
\end{flushleft}
\end{figure}

%% file: tables/item_summary_stats.tex
\begin{table}[htbp]
\centering
\caption{Item-level summary statistics}
\label{tab:item_summary_stats}
\begin{tabular}{lllll}
\toprule
                              & Mean     & SD      & Min     & Max      \\
\midrule
 List price                   & 14493.63 & 2817.51 & 7215.07 & 23193.62 \\
 Sold                         & 0.99     & 0.09    & 0.00    & 1.00     \\
 Sold with bargaining         & 0.20     & 0.40    & 0.00    & 1.00     \\
 Received an offer            & 0.39     & 0.49    & 0.00    & 1.00     \\
 Number of observable matches & 1.57     & 1.07    & 0.00    & 16.00    \\
 Sale price                   & 13995.68 & 2772.25 & 2561.87 & 23175.51 \\
 Sale price / list price      & 0.97     & 0.06    & 0.17    & 1.00     \\
 Time to sell (days)          & 19.38    & 50.09   & 0.00    & 1428.64  \\
 Time to first match (days)   & 15.95    & 40.06   & 0.00    & 373.91   \\
 Observations                 & 12451    &         &         &          \\
\bottomrule
\end{tabular}
\vspace{0.5em}
\begin{minipage}{0.9\textwidth}
\footnotesize
\textit{Notes:} The unit of observation is a listing. List and sales prices are residualized using the hedonic regression described in the main text. An observable match is defined as the arrival of a distinct buyer through a non seller comment or a purchase. Sold with bargaining indicates items that sell in a bargaining match (C) without walkaway. Received an offer indicates items with at least one observed offer. Time to sell is reported only for sold items. Time to first match is reported only for items with at least one buyer match.
\end{minipage}
\end{table}

%% file: tables/match_summary_stats.tex
\begin{table}[htbp]
\centering
\caption{Observable match-level summary statistics}
\label{tab:match_summary_stats}
\begin{tabular}{lllll}
\toprule
                                             & Mean   & SD     & Min   & Max     \\
\midrule
 Purchased at list price (A)                 & 0.473  & 0.499  & 0.000 & 1.000   \\
 An offer was made and accepted (C)          & 0.200  & 0.400  & 0.000 & 1.000   \\
 Neither of the two above (D)                & 0.327  & 0.469  & 0.000 & 1.000   \\
 An offer was made                           & 0.320  & 0.466  & 0.000 & 1.000   \\
 An offer was accepted but not purchased (W) & 0.374  & 0.484  & 0.000 & 1.000   \\
 Time to purchase after acceptance (days)    & 3.208  & 19.400 & 0.000 & 647.701 \\
 Observations                                & 19525  &        &       &         \\
\bottomrule
\end{tabular}
\vspace{0.5em}
\begin{minipage}{0.9\textwidth}
\footnotesize
\textit{Notes:} The unit of observation is an observable match. An observable match is defined as the arrival of a distinct buyer through a non seller comment or a purchase. Offer and list prices are residualized using the hedonic regression described in the text. Acceptance (C) refers to matches in which an offer is accepted by the seller. Walkaway (W) is defined only for accepted offer matches (C) and indicates that the original offer maker does not complete the purchase. Time to purchase after acceptance is measured for accepted offer matches (C) as the time from the acceptance event to the first subsequent sale of the item, regardless of whether the original buyer completes the purchase or a different buyer does so in a later match.
\end{minipage}
\end{table}

%% file: tables/commitment_walkaway_table.tex
\begin{table}[htbp]
\centering
\caption{Commitment Premium, Walkaway, and Time to Purchase}
\label{tab:commitment_walkaway_time}
\begin{tabular}{lccc}
\toprule
 & Accepted & Walkaway & $\log(1+\mathrm{time\ to\ purchase})$ \\
\midrule
Constant & 0.482 & 0.348 & 0.476 \\
 & (0.009) & (0.013) & (0.023) \\
Discount & -1.365 & 0.967 & 2.231 \\
 & (0.075) & (0.126) & (0.228) \\
Commit & 0.268 & -0.089 & -0.162 \\
 & (0.014) & (0.016) & (0.028) \\
Observations & 5913 & 3901 & 3880 \\
$R^2$ & 0.114 & 0.023 & 0.032 \\
Adjusted $R^2$ & 0.114 & 0.022 & 0.031 \\
\bottomrule
\end{tabular}
\vspace{0.5em}
\begin{minipage}{0.9\textwidth}
\footnotesize
\textit{Notes:} The unit of observation is an observable match, using the first offer in each match. Column~1 is a linear probability model for acceptance. The sample consists of observable matches in which the buyer makes an offer (C or D). Column~2 is a linear probability model for walkaway. The sample consists of accepted offer matches (C only). Column~3 regresses $\log(1+\mathrm{time\ to\ purchase})$, where time is measured in days. The sample consists of accepted offer matches (C) that are eventually purchased. The discount is defined as $(p_0-p_o)/p_0$, where $p_0$ is the residualized list price and $p_o$ is the residualized offer amount.
\end{minipage}
\end{table}

%% file: tables/table_estimation_bootstrap_se_single_r.tex
% auto-generated
\begin{table}[!htbp]
\centering
\caption{Structural estimates with bootstrap standard errors (\(r=0.05\))}
\label{tab:estimation_bootstrap_se_single_r}
\begin{threeparttable}
\begin{tabular}{lccc}
\toprule
\multicolumn{4}{l}{\textbf{Panel A: Calibrated or directly observed}} \\
\addlinespace[0.3em]
 & $N_S$ & $N_B$ & $\kappa$ \\
\cmidrule(lr){2-4}
 & \shortstack[c]{11983 \\ (33.143)} & \shortstack[c]{9650 \\ (33.240)} & \shortstack[c]{0.687 \\ (0.000)} \\
\addlinespace[0.6em]
\cmidrule(lr){1-4}
\addlinespace[0.4em]
\multicolumn{4}{l}{\textbf{Panel B: Arrival rates}} \\
\addlinespace[0.3em]
 & $\lambda_R$ & $\lambda_S$ & $\lambda_B$ \\
\cmidrule(lr){2-4}
 & \shortstack[c]{3.402 \\ (0.328)} & \shortstack[c]{0.642 \\ (0.007)} & \shortstack[c]{0.797 \\ (0.009)} \\
\addlinespace[0.6em]
\cmidrule(lr){1-4}
\addlinespace[0.4em]
\multicolumn{4}{l}{\textbf{Panel C: Valuation Distributions and Search Cost}} \\
\addlinespace[0.3em]
 & $c$ & $F_S\ (Q1, Q2, Q3)$ & $F_B\ (Q1, Q2, Q3)$ \\
\cmidrule(lr){2-4}
 & \shortstack[c]{819.989 \\ (441.592)} & \shortstack[c]{Q1: 7567.56 (205.56) \\ Q2: 8737.48 (206.19) \\ Q3: 9851.36 (254.34)} & \shortstack[c]{Q1: -25133.40 (20266.28) \\ Q2: 11099.16 (12213.27) \\ Q3: 52215.90 (6071.48)} \\
\bottomrule
\end{tabular}
\begin{tablenotes}[flushleft]
\footnotesize
\item Notes: Each cell reports the point estimate; bootstrap standard errors are in parentheses. 
Bootstrap samples are constructed by resampling items with replacement (item-level bootstrap), with all associated matches and timelines for a given item kept together. 
The number of bootstrap draws is \(B=1000\).
\end{tablenotes}
\end{threeparttable}
\end{table}

%% file: tables/welfare_unified_r0p05.tex
\begin{table}[t]
    \centering
    \begin{threeparttable}
    \begin{tabular}{lccc}
        \toprule
        & \multicolumn{1}{c}{Baseline} & \multicolumn{1}{c}{Full commitment} & \multicolumn{1}{c}{Difference} \\
        \midrule
        Total welfare & 16274.28 & 16411.92 & 137.65 \\
        \addlinespace[0.3em]
        \textbf{Sellers} & & & \\
        \quad Overall & 3020.66 & 3287.44 & 266.78 \\
        \quad Q1 & 3313.49 & 3534.77 & 221.29 \\
        \quad Q2 & 3097.27 & 3346.15 & 248.88 \\
        \quad Q3 & 2728.53 & 3108.14 & 379.62 \\
        \quad Q4 & 2381.52 & 2722.78 & 341.27 \\
        \addlinespace[0.3em]
        \textbf{Buyers} & & & \\
        \quad Overall & 14767.33 & 14572.61 & -194.72 \\
        \quad Q1 & -11249.97 & -11311.08 & -61.11 \\
        \quad Q2 & -172.10 & -353.32 & -181.22 \\
        \quad Q3 & 36190.29 & 35896.37 & -293.92 \\
        \quad Q4 & 87185.97 & 86860.95 & -325.02 \\
        \addlinespace[0.3em]
        \textbf{Platform} & & & \\
        \quad Overall & 1361.38 & 1389.05 & 27.67 \\
        \bottomrule
    \end{tabular}
    \caption{Welfare under baseline and full commitment}
    \label{tab:welfare_unified_r0p05}
    \begin{tablenotes}[flushleft]
    \footnotesize
    \item Notes: Welfare is measured on a per listing basis. Seller welfare is given by \(U_S(s)\), buyer welfare by \(U_B(b)\), and platform welfare by the scalar continuation value \(U_P\). Overall welfare is defined in \eqref{eq:overall_welfare}.
    Rows labeled ``Overall'' report welfare averaged over the full seller and buyer valuation distributions.
    Rows labeled \(Qk\) report welfare averaged over agents whose valuations lie in the \(k\)-th quartile of the relevant distribution, holding the distribution of the opposing side fixed.
    Differences are computed as Full commitment minus Baseline.
    \end{tablenotes}
    \end{threeparttable}
\end{table}

%% file: tables/action_p0_p1_baseline_diff_by_s_quartile_r0p05.tex
\begin{table}[t]
    \centering
    \begin{threeparttable}
    \begin{tabular}{lccccc ccccc}
        \toprule
        & \multicolumn{5}{c}{Baseline} & \multicolumn{5}{c}{Difference} \\
        \cmidrule(lr){2-6} \cmidrule(lr){7-11}
        & A & C & D & \(p^0\) & \(p^{1}\) & \(\Delta A\) & \(\Delta C\) & \(\Delta D\) & \(\Delta p^0\) & \(\Delta p^{1}\) \\
        \midrule
        Overall & 0.231 & 0.270 & 0.498 & 14492.5 & 13554.2 & 0.019 & -0.092 & 0.073 & 281.4 & -353.5 \\
        Q1 & 0.251 & 0.256 & 0.493 & 12804.0 & 11497.7 & 0.003 & -0.023 & 0.020 & 160.8 & 114.6 \\
        Q2 & 0.220 & 0.280 & 0.500 & 14311.6 & 13101.0 & 0.059 & -0.089 & 0.030 & 60.3 & 196.6 \\
        Q3 & 0.247 & 0.253 & 0.500 & 14773.9 & 13934.0 & -0.004 & -0.068 & 0.072 & 462.3 & 84.8 \\
        Q4 & 0.208 & 0.292 & 0.500 & 16080.4 & 15460.2 & 0.018 & -0.188 & 0.171 & 442.2 & -334.6 \\
        \bottomrule
    \end{tabular}
    \caption{Action shares and prices under baseline, with changes under full commitment}
    \label{tab:action_p0_p1_baseline_diff_r0p05}
    \begin{tablenotes}[flushleft]
    \footnotesize
    \item Notes: Rows report means overall and by quartiles of the seller valuation \(s\).
    Action shares are averaged over all buyer types \(b\), with \(C\) pooling $CN$ and $CS$.
    \(p^0\) denotes the equilibrium list price. For counteroffers, \(p^{1}\) pools \(p^1_N(s)\) for $CN$ outcomes and \(p^1_S(s,b)\) for $CS$ outcomes.
    Difference columns report Full Commitment minus Baseline.
    \end{tablenotes}
    \end{threeparttable}
\end{table}

%% file: tables/table_hedonic.tex
\begin{table}[!htbp] \centering
  \caption{Hedonic price regression}
\begin{tabular}{@{\extracolsep{5pt}}lc}
\\[-1.8ex]\toprule \\[-1.8ex]
\\[-1.8ex] & \multicolumn{1}{c}{List price}  \\
\\[-1.8ex] & (1) \\
\midrule \\[-1.8ex]
 Constant & 14584.538 \\
& (454.439) \\
 Condition: Some scratches & -1540.795 \\
& (71.508) \\
 Condition: Scratched & -3437.442 \\
& (87.886) \\
 Shipping payer: Seller pays shipping & -1665.865 \\
& (548.765) \\
 Shipping duration: Ships in 2--3 days & -127.359 \\
& (75.329) \\
 Shipping duration: Ships in 4--7 days & -156.073 \\
& (132.291) \\
\\[-1.8ex]
 Observations & 13391 \\
 $R^2$ & 0.117 \\
 Adjusted $R^2$ & 0.113 \\
 Residual Std. Error & 3663.424 (df=13324) \\
 F Statistic & 26.786 (df=66; 13324) \\
\bottomrule \\[-1.8ex]
\addlinespace
\multicolumn{2}{p{0.95\textwidth}}{\footnotesize \textit{Notes:} The specification includes prefecture fixed effects and shipping-method fixed effects. The intercept corresponds to the reference categories for all covariates (no visible scratches, buyer paid shipping, shipping within 1--2 days, the platform default shipping option, and the baseline prefecture).} \\
\end{tabular}
\end{table}

%% file: tables/table_estimation_summary_panelc_only.tex
% auto-generated
\begin{table}[!htbp]
\centering
\caption{Structural parameters varying with discount rates}
\label{tab:estimation_summary_panelc_only}
\begin{threeparttable}
\begin{tabular}{lccccccc}
\toprule
 & & \multicolumn{3}{c}{$F_S$} & \multicolumn{3}{c}{$F_B$} \\
\cmidrule(lr){3-5}\cmidrule(lr){6-8}
$r$ & $c$ & Q1 & Q2 & Q3 & Q1 & Q2 & Q3 \\
\addlinespace[0.2em]
0.020 & 728.760 & 5643.77 & 6645.13 & 7263.40 & -64536.50 & 18762.92 & 113291.17 \\
0.050 & 819.989 & 7567.56 & 8737.48 & 9851.36 & -25133.40 & 11099.16 & 52215.90 \\
0.070 & 457.078 & 8004.47 & 9197.49 & 10389.48 & -17659.87 & 9620.86 & 40579.05 \\
0.120 & 284.943 & 8523.38 & 9734.55 & 11001.46 & -9918.86 & 8055.64 & 28453.11 \\
0.140 & 240.449 & 8639.82 & 9853.88 & 11135.55 & -8382.19 & 7735.98 & 26026.89 \\
0.210 & 181.011 & 8893.42 & 10112.91 & 11424.39 & -5331.59 & 7083.82 & 21172.83 \\
0.290 & 273.152 & 9050.24 & 10272.73 & 11602.15 & -3670.09 & 6712.03 & 18493.66 \\
\bottomrule
\end{tabular}
\begin{tablenotes}[flushleft]
\footnotesize
\item Notes: This table reports only the structural parameters that vary with the discount rate \(r\). 
The column labeled \(c\) reports the estimated flow search cost. Columns labeled \(F_S\) and \(F_B\) report the first, second, and third quartiles of the seller and buyer valuation distributions, respectively.
\end{tablenotes}
\end{threeparttable}
\end{table}

%% file: draft/online_appendix.tex
\section{Dataset Construction}\label{sec:dataset_construction}

\subsection{Selection of Subsample}

Mercari provided data covering over 400 million items listed during the year from 2019 to 2021.
I first restrict attention to items listed in the
"Smartphones" category within "Mobile Phones,"
which is nested under "Home Appliances, Smartphones, and Cameras."
This yields 763{,}079 items.

I then restrict the sample to iPhones.
Items are classified as iPhones if either the brand name is recorded as Apple or,
because the brand field is frequently missing,
the item title or description contains the string "iPhone,"
allowing for common orthographic variants.
This restriction leaves 414{,}363 items.

Next, I infer the iPhone model from the item title or description text.
iPhone~7 is the most frequently traded model in the data,
resulting in a subsample of 61{,}132 items.
I similarly infer storage capacity from the title or description.
Among iPhone~7 listings, 128GB is the most common storage size,
with 32{,}030 items.

I then examine the distribution of list prices by item condition
within the iPhone~7 128GB subsample.
As shown in Figure~\ref{fig:price_dist},
the price distributions for conditions "No noticeable damage", "Some damage", and "Damaged" are similar to each other
and clearly distinct from the remaining conditions.
I therefore restrict attention to items in these three item conditions for the main analysis, leaving me with 28,230 items.

\begin{figure}[t]
    \centering
    \includegraphics[width=0.9\linewidth]{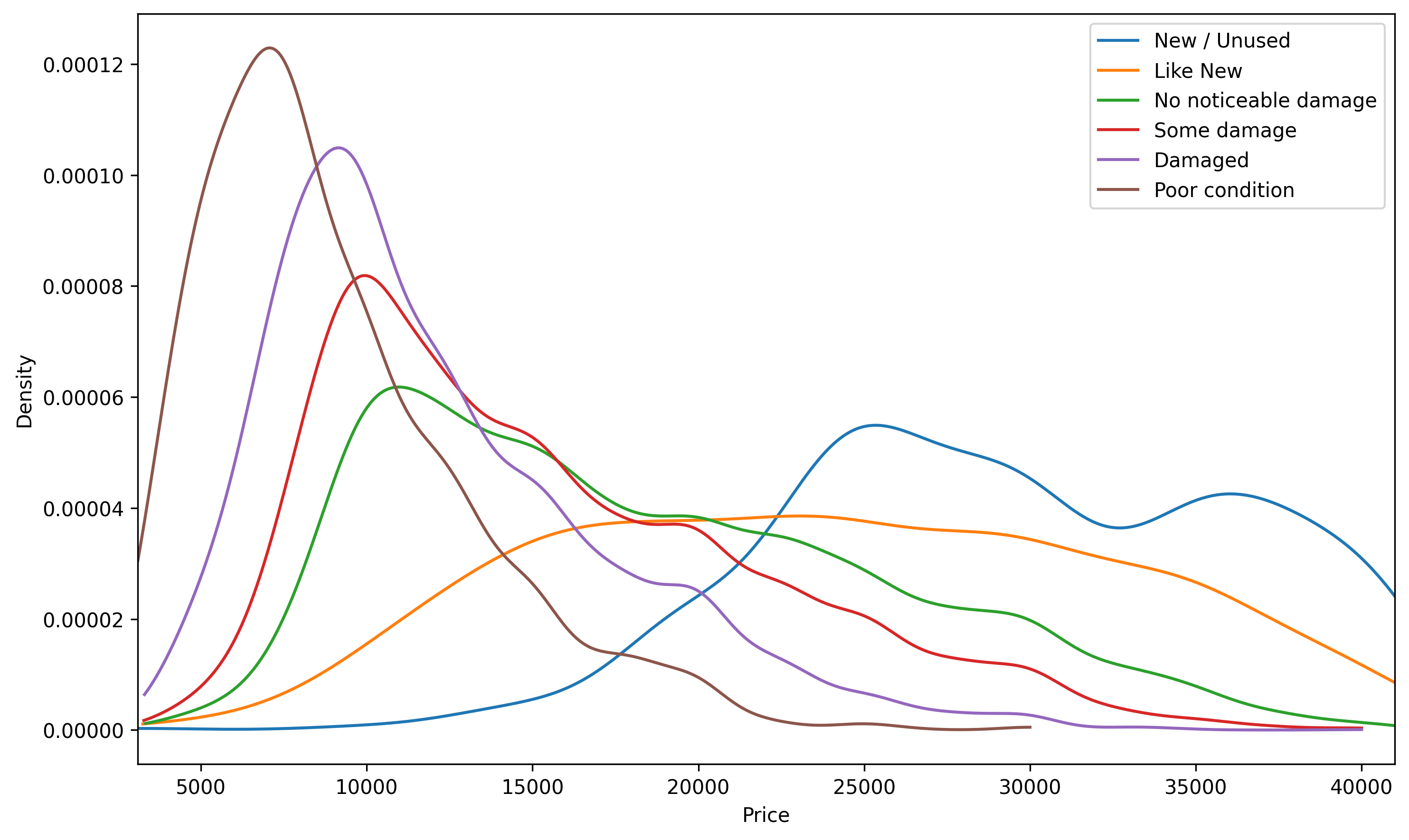}
    \caption{Distribution of prices by item condition}
    \label{fig:price_dist}
    \vspace{0.4em}
    \begin{minipage}{0.9\textwidth}
    \footnotesize
    \textit{Notes:} The figure plots kernel density estimates of listing prices for the iPhone~7 128GB subsample, separately by the platform reported item condition. For visual clarity, prices are trimmed to the 1st and 99th percentiles of the price distribution within this subsample prior to density estimation.
    \end{minipage}
\end{figure}

I further restrict the sample to items listed on or after June 24, 2020, when price change records begin to be consistently available, leaving 13,397 items.

\subsection{Cleansing of the Dataset}

I extract offer amounts from comments as follows. For each comment, I first normalize numeric expressions and parse them into numbers. I then retain values between 5,000 and 30,000 and exclude any values that coincide with previously observed list prices, price changes, or earlier offers for the same item. If multiple candidates remain, I select one as the offer amount.

The seller can change its list price even without receiving an offer from a buyer. I therefore treat the price after any price changes that occurred before any offer from a buyer as the effective list price. I define a match as the first arrival of a distinct buyer, identified by a non-seller comment or purchase. An item is flagged as having multiple-match buyers if a buyer reappears after another buyer has arrived. %This occurs for 8.5\% (1136/13397) of items.

I label seller and buyer actions at the match level as follows. An offer is classified as accepted if the seller subsequently updates the posted price to exactly the offer amount after the offer is made. Using this information, each buyer-item match is assigned an action type: $C$ if the match includes an accepted offer, $A$ if the item is sold without any offer from that buyer, and $D$ otherwise. For matches of type $C$, I further classify buyer behavior by defining a walkaway indicator equal to one if the offering buyer does not complete the purchase. Finally, I label buyer commitment by setting a commit indicator equal to one if any comment within a $C$-type match contains keywords indicating an intention to complete the transaction promptly. The list of keywords is shown in Table \ref{tab:commit_words}.

\input{tables/commit_words}

The analysis starts from a sample of 13,397 items. I first drop 6 items with missing residualized list prices, leaving 13,391 items. I then exclude 204 items for which either the residualized offer amount or the residualized sales price exceeds the residualized list price, reducing the sample to 13,187 items. Finally, I trim 736 items that fall outside the 1st - 99th percentile range of the residualized list price and the time to first match, yielding a final analysis sample of 12,451 items. In total, 946 items are removed across all trimming steps.

\section{Simulation of the Equilibrium}\label{sec:vfi}

I compute the stationary equilibrium by value function iteration on discretized state spaces, taking the estimated primitives and calibrated parameters as given.

\paragraph{Discretization.}
I use \(100\) grid points each for seller values \(s\) and buyer values \(b\), constructed by drawing from the estimated distributions and sorting the draws.
I discretize feasible prices on a uniform grid with \(200\) points over \([0,100000]\).
I fix \((t,r,c,\lambda_S,\lambda_B,\lambda_R,\kappa)\) and \((N_S,N_B)\) at their estimated or calibrated values.
In the full commitment counterfactual, I set \(c=\infty\) and hold all other primitives fixed.
I initialize \(U_S(s)\), \(U_B(b)\), and \(V_B(s,b)\).

\paragraph{Policy updates given \((U_S,U_B,V_B)\).}
For each iteration, I update the equilibrium policies implied by the current value functions.

First, I compute the committed offer \(p^1_N(s)\) from \eqref{eq:committed_offer}.
Second, I compute the non-committed offer \(p^1_S(s,b)\) by searching over the discretized price grid for the smallest \(p^1\) that satisfies the seller acceptance condition \eqref{eq:seller_accept_noncommit}.
This step uses the walkaway cutoff \(s^*(b,p^1)\) defined by \(V_B(s^*, b) = b - p^1\) and the induced walkaway probability \(F_S(s^*(b, p^1))\).

Given \(\{p^1_N,p^1_S\}\), for each candidate list price \(p^0\) I evaluate buyer continuation values for \(\chi\in\{A,CN,CS,D\}\) using the payoff expressions in \eqref{eq:vbchi} and choose the maximizing action \(\chi_B(s,b;p^0)\).
Finally, I compute the seller optimal list price \(p^0(s)\) by maximizing the objective in \eqref{eq:seller_maximization}, integrating over the \(100\) point buyer grid and the induced buyer actions.

\paragraph{Value updates, stabilization, and convergence.}
Using the updated policies, I update \(U_S\), \(U_B\), and \(V_B\) using their definitions in Section~\ref{sec:model}.
For numerical stability, I smooth the updated value functions at each iteration using a Gaussian kernel with bandwidth 5 on the discretized state grids.
I iterate until the relative change in both \(U_S\) and \(U_B\) falls below \(10^{-3}\).

%\end{document}

%% file: tables/commit_words.tex
\begin{table}[t]
\centering
\caption{Commitment-related expressions}
\label{tab:commit_words}
\begin{tabular}{llr}
\toprule
Japanese & English meaning & Count \\
\midrule
sokketsu & Immediate decision / buy outright & 625 \\
soku kounyuu & Buy immediately & 520 \\
sugu & Immediately / right away & 353 \\
shiharai & Payment & 147 \\
oshiharai & Payment & 60 \\
hayame & Early / sooner & 49 \\
honjitsu chuu & By today & 42 \\
nyuukin & Payment made / remittance & 39 \\
hayaku & Quickly & 32 \\
kyou chuu & Within today & 20 \\
sugu ni & Right away & 19 \\
isogi & Urgent & 17 \\
gozen chuu & By the morning & 15 \\
sokujitsu & Same day & 13 \\
kanarazu & Definitely / for sure & 13 \\
isoide & In a hurry & 7 \\
ashita chuu & By tomorrow & 7 \\
kettei & Decision & 6 \\
saitan & As soon as possible & 5 \\
kimeru & Decide & 2 \\
\bottomrule
\end{tabular}
\vspace{0.5em}
\begin{minipage}{0.9\textwidth}
\footnotesize
\textit{Notes:} The table reports the number of accepted-offer matches (action $C$) in which each expression appears in the negotiation comments. The unit of observation is an observable match, defined as the arrival of a distinct buyer through a non seller comment or a purchase. Each expression is counted at most once per match, regardless of how many times it appears within the same negotiation.
\end{minipage}
\end{table}